\documentclass[a4paper,UKenglish,cleveref, autoref, thm-restate]{lipics-v2021}
\usepackage{amsmath}
\usepackage{amsfonts}
\usepackage{amsthm}
\usepackage{tikz-cd}
\usepackage{stmaryrd}
\usepackage{bussproofs}
\usepackage{color}
\usepackage{mathpartir}
\usepackage{hyperref}
\usepackage{latexsym}

\newtheorem*{terminology}{Terminology}

\newcommand{\grD}{{D^{\kappa}}}

\newcommand{\ciD}{D}
\newcommand{\Dforall}{D^\forall}

\DeclareMathOperator{\nextop}{next}
\DeclareMathOperator{\now}{now}
\DeclareMathOperator{\step}{step}

\DeclareMathOperator{\ar}{ar}

\DeclareMathOperator{\tlater}{\blacktriangleright}

\DeclareMathOperator{\PSh}{PSh}

\newcommand{\iso}{\cong}
\newcommand{\defeq}{\mathbin{\overset{\textsf{def}}{=}}}

\newcommand{\later}{\triangleright}

\newcommand\tabs[2]{\lambda (#1\! :\! #2).}

\newcommand\tickc{\diamond}

\newcommand\tapp[2][\tickA]{#2\,[#1] }

\newcommand{\tickA}{\alpha}
\newcommand{\tickB}{\beta}
\newcommand\latbind[2]{{\triangleright}\, (#1 \!: \!#2) .}
\newcommand\latbindsmall[2]{{\triangleright}\, (#1 : #2) .}
\newcommand\toksubst[3][\kappa]{\left[#2/#3\right]}

\newcommand{\capp}[2][\kappa]{#2\,[#1]}

\newcommand{\clocktype}{\mathsf{clock}}

\newcommand{\fix}[1][\kappa]{\mathsf{fix}^{#1}}

\newcommand{\Set}{\mathsf{Set}}
\newcommand{\opcat}[1]{{{#1}^{\mathrm{op}}}}
\newcommand{\slice}[2]{#1/{#2}}

\newcommand{\Prop}{\mathsf{Prop}}
\newcommand{\N}{\mathbb{N}}

\newcommand\hastype[4][]{#2 \vdash_{#1} #3: #4}

\newcommand\wfcxt[2][]{#2 \vdash_{#1}}
\newcommand\istype[3][]{\ensuremath{#2 \vdash_{#1} #3 \, \operatorname{type}}}

\newcommand\genericjudg[3][]{#2 \vdash_{#1} #3}

\newcommand{\subst}[2]{[#1/#2]}

\newcommand{\peq}{=}
\newcommand\sym[1]{\mathsf{#1}}

\newcommand{\tirr}[1][\kappa]{\sym{tirr}^{#1}}

\newcommand{\force}{\sym{force}}

\newcommand{\idty}[3]{#2 \peq_{#1} #3}

\newcommand{\Pfin}{{\mathsf{P}_{\mathsf{f}}}}
\newcommand{\Dfin}{{\mathsf{D}_{\mathsf{f}}}}

\newcommand{\IsEquiv}[1]{\mathsf{IsEquiv}(#1)}
\newcommand{\univ}[1]{\ensuremath{\operatorname{U}_{#1}}}
\newcommand{\elems}[1]{\ensuremath{\operatorname{El}_{#1}}}
\newcommand{\univin}[2]{\mathsf{in}_{#1,#2}}%{\ensuremath{\operatorname{in}_{#1,#2}}}

\newcommand{\forallcode}[1]{\ensuremath{\overline{\forall}}}

\newcommand{\latbindcode}[3]{\overline{\triangleright}\, (#1:#2) . #3}

\newcommand{\latercode}[1][\kappa]{\overline{\triangleright}^{#1}}

\newcommand{\syncode}[1]{\overline{#1}}

\newcommand{\peqcode}[2]{ #1 = #2}

\newcommand{\prop}[1]{\ensuremath{\operatorname{Prop}_{#1}}}

\newcommand{\similar}{\lesssim}

\newcommand{\Func}[2]{\mathsf{Func}(#1, #2)}

\newcommand{\catT}[1][\indexord]{\mathbb{T}_{#1}}

\newcommand{\indexord}{\rho}

\newcommand{\triple}[3]{(#1;#2;#3)}
\newcommand{\quadruple}[4]{(#1;#2;#3;#4)}
\newcommand{\sixtuple}[6]{(#1;#2;#3;#4;#5;#6)}
\newcommand{\seventuple}[7]{(#1;#2;#3;#4;#5;#6;#7)}
\newcommand{\timeobj}[2]{(#1;#2)}
\renewcommand{\i}{\iota}

\newcommand{\clk}{\mathrm{Clk}}

\newcommand{\Elsem}[1]{\mathcal{E}l^{#1}}
\newcommand{\Usem}[1]{\mathcal{U}^{#1}}

\newcommand{\PElsem}[1]{\mathcal{E}l_{\mathsf{Prop}}^{#1}}
\newcommand{\Propsem}[1]{\mathcal{U}_{\mathsf{Prop}}^{#1}}

\newcommand{\code}[1]{\ulcorner #1 \urcorner}%{\ensuremath{\overline{#1}}}

\newcommand{\restrict}[2]{{#1}|_{#2}}
\newcommand{\compr}[2]{#1.#2}

\newcommand{\from}{\leftarrow}

\newcommand{\true}{\mathsf{true}}
\newcommand{\false}{\mathsf{false}}

\newcommand{\kapsubset}[1][\kappa]{\subseteq_{#1}}

\newcommand{\colim}{\mathsf{colim}}

\newcommand{\inv}[1]{#1^{-1}}

\newcommand{\FSA}{\mathcal{E}}
\newcommand{\FSB}{\mathcal{E'}}

\renewcommand{\vartheta}{\delta}

\newcommand{\grtotal}{\mathsf{GR}}
\newcommand{\gr}[1]{\mathsf{GR}[#1]}

\newcommand{\ccat}{\mathbb{C}}
\newcommand{\dcat}{\mathbb{D}}

\newcommand{\op}{\mathsf{op}}
\newcommand{\Tm}{\mathsf{Tm}}
\newcommand{\fv}[1]{\mathsf{fv}(#1)}

\title{Multi-clocked Guarded Recursion Beyond $\omega$} %TODO Please add

\author{Rasmus Ejlers Møgelberg}{IT University of Copenhagen, Denmark}{mogel@itu.dk}{https://orcid.org/0000-0003-0386-4376}{This work was supported by the Independent Research Fund Denmark grant number
2032-00134B.}

\authorrunning{Rasmus Ejlers Møgelberg} %TODO mandatory. First: Use abbreviated first/middle names. Second (only in severe cases): Use first author plus 'et al.'

\Copyright{Rasmus Ejlers Møgelberg} %TODO mandatory, please use full first names. LIPIcs license is "CC-BY";  http://creativecommons.org/licenses/by/3.0/

 %TODO mandatory; please add comma-separated list of keywords

\category{} %optional, e.g. invited paper

\relatedversion{} %optional, e.g. full version hosted on arXiv, HAL, or other respository/website
\acknowledgements{I thank Ji{\v{r}}{\'\i} Ad{\'a}mek for discussions of the material of Section~\ref{sec:semantic:condition}.}%optional

\nolinenumbers %uncomment to disable line numbering


\EventEditors{Fredrik Nordvall Forsberg and James McKinna}
\EventNoEds{2}
\EventLongTitle{31st International Conference on Types for Proofs and Programs (TYPES 2025)}
\EventShortTitle{TYPES 2025}
\EventAcronym{TYPES}
\EventYear{2025}
\EventDate{June 9--13, 2025}
\EventLocation{University of Strathclyde, Glasgow, Scotland}
\EventLogo{}
\SeriesVolume{384}
\ArticleNo{11}
\begin{document}

\maketitle

\begin{abstract}
Type theories with multi-clocked guarded recursion provide a flexible framework for programming with coinductive types encoding productivity in types. Combining this with solutions to general guarded domain equations one can also construct relatively simple denotational models of programming languages with advanced features. These constructions have previously been explored in the setting of extensional type theory through a presheaf model, which proves correctness of encodings of W-types. That model has been adapted to presheaves of cubical sets (functors into the category of cubical sets), where the model verifies correctness of encodings also of coinductive types whose definitions involve quotient inductive types such as finite powersets or finite distributions. Likewise the cubical model also verifies correctness of coinductive predicates defined using existential quantification
and allows the results to be related to the global world of cubical sets. 

This paper looks at how to extend the extensional presheaf model of multi-clocked guarded recursion to higher ordinals, so that correctness of encodings of coinductive types can be extended from W-types to those involving finite powersets and finite distributions, as well as coinductive predicates involving existential quantification. This extension will allow results previously proved in Clocked Cubical Type Theory to be interpreted in a model based on set-theory, proving the correctness of these results as understood in their usual set theoretic interpretation. 
\end{abstract}

\section{Introduction}
\label{sec:introduction}
%%
%% Bibliography
%%

The complex numbers form a mathematical paradise in which all polynomials have roots. However, they are not just a dream world without connection to 
the more earthly world of the real numbers. For example, 
%Cardano\footnote{First published in his book \emph{Artis Magnae, Sive de Regulis Algebraicis}, 1545} 
%noticed that 
computing with complex numbers sometimes 
leads to the construction of real solutions to polynomial equations. In the world of programming language semantics the equations of interest are
recursive domain equations $X \iso F(X)$. Many of these do not have solutions in the category of sets, but the topos of trees~\cite{ToT} serves as a mathematical paradise
in which guarded versions of these equations $X \iso F(\later X)$ all have solutions. One can therefore construct models of advanced programming languages using 
the topos of trees, but can one then relate these results to the real world? 

The answer depends on the type of results, and what is meant by the `real world'. For example, the original paper on the topos of trees~\cite{ToT}
constructs a model of a programming language with higher order store in the internal language of the topos of trees ($\Set^\opcat\omega$), 
and then argues how these results relate to the real world of the category of sets. In some cases, this is not possible, however. 
For example, a proof of an existential quantification of a predicate in the topos of trees does not always externalise
to a proof of an existential in sets. For this reason, Transfinite Iris~\cite{TransfiniteIris} uses higher 
indexing than $\omega$ to reason about liveness properties. Likewise, as we shall explain below, passing to the real world does not preserve functors 
like $\Pfin$ and $\Dfin$ used to model non-determinism and probabilistic choice, and this causes problems when relating models of languages 
with these constructs to the real world. 

%
%Also when modelling non-deterministic of probabilistic choice the relation to the real
%world can be lost, essentially because functors like the finite powerset functor $\Pfin$ or the finite distribution functor $\Dfin$ 
%do not preserve $\opcat{\omega}$-indexed limits. 

The problem can be stated type theoretically using \emph{multi-clocked guarded recursion} as in 
Clocked Type Theory (CloTT)~\cite{bahr2017clocks} and its 
presheaf model~\cite{clottmodel}. In CloTT the delay modality $\later^\kappa$ 
is indexed by a clock $\kappa$, and one can quantify these clocks universally. 
The model contains a copy of the standard set theoretic model of type theory as well as a copy of the topos of trees, 
and universal quantification acts as a form of externalisation, relating results proved in the topos of tree to those in 
the real world of set theory. 
%
%one can pass from the paradise of the topos of trees to the real world of set theory by universally quantifying a clock. 
In this setting, the problem of concluding from existential quantification in the topos of 
trees to actual existence in sets can be expressed as commutativity of quantification over clocks with existential quantification. 

Multi-clocked guarded recursion also offers a way of programming with coinductive
types and predicates in type theory, using guarded recursion to express the requirement of productivity in types~\cite{atkey13icfp,GDTT}. 
Specifically, for functors $F$ commuting with clock quantification, the solution to the guarded recursive domain equation 
$X(\kappa) \iso F(\later^\kappa X(\kappa))$ externalises to a terminal coalgebra with carrier
$\forall\kappa . X(\kappa)$ in the real world of set theory. %
%Their observation was that if $F$ is a functor, then one can construct the terminal coalgebra for $F$
%by first solving $X(\kappa) \iso F(\later^\kappa X(\kappa))$, then constructing the terminal coalgebra as $\forall\kappa . X(\kappa)$.
%This works if $F$ commutes with universal quantifcation over clocks. For this to be useful, we need a large supply of type and predicate
%constructors that commute with clock quantification. 
Unfortunately, standard functors such as finite powerset functor $\Pfin$ 
or the finite distribution functor $\Dfin$ do not commute with with clock quantification in the standard model, and so one cannot 
use these tools to program with non-deterministic or probabilistic processes. 
%Likewise, since existential quantification does not commute 
%with clock quantification, one cannot use guarded recursion to reason about bisimilarity. 

By varying the notion of `real world', however, the situation can be improved. Kristensen et al.~\cite{CubicalCloTT} 
construct a model of the combination of CloTT and 
Cubical Type Theory~\cite{CTT} by considering presheaves valued in cubical sets. In this model, existential quantification 
%in the extended world 
%\emph{does} correspond to existential quantification in Cubical Type Theory, 
and functors like $\Pfin$ and $\Dfin$, constructed as 
quotient inductive types on the category of hsets \emph{do} commute with clock quantification. This means that guarded recursion can be used 
as a way of programming with coinductive types and predicates in Cubical Type Theory, including those involving $\Pfin$ and $\Dfin$ 
and existential quantification. Clocked Cubical Type Theory has been used to model programming languages with non-determinism~\cite{mogelberg2021two} and
probabilistic choice~\cite{POPL25}. These models are constructed by solving guarded recursive domain equations
(also with negative occurences), and using the commutativity results above, the models can be related to the
operational semantics of these languages, defined as coinductive non-deterministic and probabilistic processes. Likewise, 
Clocked Cubical Type Theory has been used to prove existence and non-existence of distributive laws between the coinductive delay monad and other monads~\cite{FewExtraPages}, and 
to model programming languages with gradual typing~\cite{GiovanniniDN25}. Some of these results have been formalised in Guarded Cubical 
Agda~\cite{GuardedCubicalAgda}, an extension of the Cubical Agda proof assistant~\cite{CubicalAgda}.

%
%Therefore, one can use guarded recursion to construct 
%models of advanced programming languages and for coalgebraic programming REFs, and the results thus proved 
%can be related to the world of cubical sets. This last step is never done in the examples referred to, but we will argue for that in this paper. \todo{Do that!}

The question is now how tied these results are to cubical type theory. For example, the results proved about existence and non-existence of 
distributive laws between monads as interpreted in the model of Kristensen et al. are results about these monads as seen as monads of hsets in the 
cubical model. Do the same results hold for these monads on $\Set$? Likewise, results about contextual equivalence of programs proved using 
the cubical model a priori only say something about their semantics as read in cubical sets. How do these relate to the ordinary set theoretic
interpretations? 

In this paper we return to the set theoretic model of CloTT and look at how varying the indexing ordinal used in the model construction implies
commutativity of clock quantification with certain existential quantifications and with functors constructed as quotient inductive types
such as $\Pfin$ and $\Dfin$. The contributions are as follows
\begin{itemize}
\item We analyse existing applications of Clocked Cubical Type Theory~\cite{CubicalCloTT,POPL25,FewExtraPages} and 
isolate requirements on the model that will allow for these to be interpreted in the set theoretic model.
\item We state the conservativity of the model of Clocked Type Theory over the standard set theoretic model of type theory. This is an
easy consequence of the model construction of~\cite{clottmodel}, but has not previously been stated explicitly. We show how this allows us
to relate results proved in CloTT to the real world of set theory.
\item We give semantic conditions for monads on sets to commute with clock quantification in the model of Clocked Type Theory
indexed by regular ordinals, and we give a simple condition on algebraic theories to generate such monads. This accounts for the functors
$\Pfin$ and $\Dfin$ mentioned above, but also monads of higher arity, such as the countable powerset monad.
\item We give conditions allowing existential quantification to commute with quantification over clocks in the model
\end{itemize}
As a result we see that the applications mentioned above can all be expressed in the internal language of a variant of the set theoretic model
indexed by a sufficiently large ordinal. We hope that these results 
%The hope is that the results of this paper 
can form part of the semantic
basis for formalisation of applications as those mentioned above also in proof assistants
not based on cubical type theory. 

We remark that the metatheory in which the models are constructed is classical set theory. 
In particular, the proof of Theorem~\ref{thm:semantic:condition} is very classical. 

\subsection*{Overview. } We start by recalling Clocked Type Theory and the encoding of coinductive
types from guarded recursive types in Sections~\ref{sec:cctt} and~\ref{sec:encoding:coind}. 
Section~\ref{sec:applications} recalls the applications of Clocked Cubical Type Theory and derives the requirements
needed to express these in the version of Clocked Type Theory based on the set theoretic model. Section~\ref{sec:model}
recalls the presheaf model of Clocked Type Theory and describes the conservativity result over the set theoretic model. 
Section~\ref{sec:alg:types} gives conditions on finitary monads to commute with clock quantification as well
as the syntactic condition on algebraic theories ensuring this. Section~\ref{sec:existential} considers commutativity of existential
quantification and clock quantification, and Section~\ref{sec:conclusion} concludes.

% \todo{Why type theory, not standard logic (ticks), using classical logic (maybe in conclusion/future work)}

%\todo{Introduce terms of SGDT and multiclocked GR above}

\section{Clocked Type Theory}

\label{sec:cctt}

\begin{figure}
\begin{mathpar}
  \inferrule*
  {\,}
  {\wfcxt{\Gamma, \kappa : \clocktype}{}}
  \and
  \inferrule*
  {\kappa : \clocktype \in \Gamma}
  {\wfcxt{\Gamma, \tickA : \kappa}{}}
  \and
  \inferrule*
  {\hastype{\Gamma}{t}{\latbind{\tickA}{\kappa} A}\\ \wfcxt{\Gamma,\tickB:\kappa,\Gamma'}}
  {\hastype{\Gamma,\tickB: \kappa,\Gamma'}{\tapp[\tickB] t}{A\toksubst{\tickB}{\tickA}}}
  \and
  \inferrule*
  {\hastype{\Gamma,\tickA:\kappa}{t}{A}}
  {\hastype{\Gamma}{\tabs{\tickA}{\kappa} t}{\latbind{\tickA}{\kappa} A}}
  \and
    \inferrule*
  {\hastype{\Gamma,\kappa : \clocktype}{t}{A}}
  {\hastype{\Gamma}{\Lambda\kappa. t}{\forall \kappa . A}}
  \and
  \inferrule*
  {\hastype{\Gamma}{t}{\forall \kappa . A}\\
    \kappa' : \clocktype\in \Gamma}
  {\hastype{\Gamma}{t [\kappa']}{A \subst{\kappa'}{\kappa}}}
 \and
  \inferrule*
  {\kappa : \clocktype \in \Gamma \\ \istype\Gamma A}
  {\hastype\Gamma{\fix}{(\later^\kappa A \to A) \to A}}
 \and
  \inferrule*
  {x : A \in \Gamma}
  {\hastype\Gamma xA}
\end{mathpar}
\caption{Selected typing rules for Clocked Type Theory~\cite{bahr2017clocks,clottmodel}.}
\label{fig:later:typing}
\end{figure}

This section gives a brief introduction to Clocked Type Theory. We refer readers to \cite{clottmodel} for full details. 

Clocked Type Theory is a dependently typed type theory with $\Pi$- and $\Sigma$-types, extensional
identity types, binary and nullary sums and products, as well as a modality $\later$ for delays and a guarded fixed point operator. 
Delays are associated with clocks, so that $\later^\kappa A$ can be understood as a type of data of type $A$ delayed by one time step
on clock $\kappa$. To program with clocks, the type theory has a special form of assumptions $\kappa : \clocktype$ that can be placed in a context.
Note that $\clocktype$ is not a type, but rather a pretype, similar to the status of the interval in Cubical Type Theory~\cite{CTT}. Introduction
and elimination for $\later$ is by abstracting and applying ticks on clocks, following the 
Fitch-style rules for programming with modal types~\cite{clouston2018fitch,drat}. Selected typing rules can be found 
in Figure~\ref{fig:later:typing}. The notation
$\later^\kappa A$ is used for $\latbind\tickA\kappa A$ when the tick $\tickA$ does not appear free in $A$. Note that variables can be introduced 
from anywhere in a context, also if they are behind ticks. This means that one can define a term $\nextop^\kappa \defeq \lambda x. \tabs\tickA\kappa x$
of type $A \to \later^\kappa A$. Given $f : \later^\kappa A \to A$, the fixed point operator computes the unique fixed point of $f \circ \nextop^\kappa$. 
The language of ticks also allows one to define a dependent version of an applicative functor action for $\later$:
\[
 \lambda f .\lambda x. \tabs\tickA\kappa{\tapp f (\tapp x)}: \later^\kappa(\Pi(x:A). B(x)) \to \Pi(x : \later^\kappa A) . \latbind\tickA\kappa{B(\tapp x)}
\]
This is, for example, useful for reasoning about elements of modal type. It can also be
used to express an extensionality principle for $\later$:
\[
(t =_{\latbindsmall\tickA\kappa A} u) \iso \latbind\tickA\kappa (\tapp t =_A \tapp u)
\]
For these applications, it is important 
that the same ticks can be used on both the logical and term side, and this is one benefit of using 
type theory rather than a language with a clear separation between term language and logic, such as topos logic. 

\subsection{Axioms}

One of the main purposes of quantification over clocks it to allow for a controlled elimination of $\later$. In Clocked Type Theory this is 
done by an operation of application to a tick constant $\tickc$, with very special typing rules, 
which can be used to construct a term of type 
\[
  \force : (\forall\kappa . \later^\kappa A) \to \forall\kappa . A.
\]
This, together with the 
\emph{tick irrelevance axiom} 
\begin{equation}\label{eq:tirr}
\inferrule*
{\hastype{\Gamma}{t}{\later^\kappa A}}
{\tirr t :  \latbind{\tickA}{\kappa}  \latbind{\tickA'}{\kappa} \idty{A}{\tapp{t}}{\tapp[\tickA'] t}}
\end{equation}
allows one to prove the following 
%We omit the details here, and just note
%that this allows to prove an isomorphism 
\begin{equation} \label{eq:force}
 \IsEquiv{\lambda (x : \forall\kappa . A) . \lambda \kappa . \tabs\tickA\kappa{\capp x}}
% \forall\kappa . \later^\kappa \!A \iso \forall\kappa. A
\end{equation}
where $\IsEquiv{-}$ is defined precisely as in homotopy type theory~\cite{hottbook}. In the present case of extensional
type theory, this axiom just states that the canonical map 
$(\forall\kappa .A) \to \forall\kappa.\later^\kappa A$ 
is an isomorphism. 

The axiom of \emph{clock irrelevance} states
\begin{equation} \label{eq:cirr}
\IsEquiv{\lambda x : A . \lambda\kappa . x}
%\inferrule*
% {\istype\Gamma A}% \\ \kappa \notin \Gamma}
% {\hastype\Gamma{\cirr}{\IsContr{\lambda x : A . \tabs\tickA\kappa x}}}
\end{equation}
Again, in our setting, this just states that the canonical map 
$A \to \forall\kappa.A$ is an isomorphism when $\kappa$ is not 
free in $A$. 
%The third is \emph{tick irrelevance}:
%\begin{equation}\label{eq:tirr}
%\inferrule*
%{\hastype{\Gamma}{t}{\later^\kappa A}}
%{\tirr t :  \latbind{\tickA}{\kappa}  \latbind{\tickA'}{\kappa} \idty{A}{\tapp{t}}{\tapp[\tickA'] t}}
%\end{equation}
In Clocked Cubical Type Theory~\cite{CubicalCloTT} 
clock irrelevance is not an axiom, but can be proved for many types. The other axioms mentioned can all 
be given computational content in Clocked Cubical Type Theory, but it is unclear how to generalise this to 
type theories not based on Cubical Type Theory. In this paper we are mostly concerned with properties 
of the model, and so we will assume clock irrelevance and tick irrelevance as axioms, deriving (\ref{eq:force}) as a consequence. 
% because the model verifies.
%\todo{Maybe discuss computational properties further at a later point?}

\subsection{Universes}

\begin{figure}
\begin{mathpar}
 \inferrule*{\Delta = \{\kappa_1, \dots, \kappa_n\} \and \hastype{\Gamma}{\kappa_i}\clocktype \text{ for }i=1,\dots, n
 }
 {\istype{\Gamma}{\univ{\Delta}}} \and
 \inferrule*{
 \hastype{\Gamma}{t}{\univ{\Delta}}
 }
 {\istype{\Gamma}{\elems{\Delta}(t)}} 
 \and
 \inferrule*{
 \hastype{\Gamma}{t}{\univ{\Delta}} \\  \Delta \subseteq \Delta'
 }{
 \hastype{\Gamma}{\univin{\Delta}{\Delta'}(t)}{\univ{\Delta'}}}
 \and
   \inferrule*{
 \hastype{\Gamma,\kappa : \clocktype}A{\univ{\Delta,\kappa}}
 }
 {
  \hastype{\Gamma}{\forallcode{\Delta}\kappa . A}{\univ{\Delta}}
 } \and
 \inferrule*{
 \hastype{\Gamma, \tickA : \kappa}{A}{\univ{\Delta}} \and \kappa \in \Delta
 }{
 \hastype{\Gamma}{\latbindcode\tickA\kappa A}{\univ{\Delta}}
 }
\end{mathpar}
\caption{Selected typing rules for universes.}
\label{fig:universes}
\end{figure}

It is inconsistent with the clock irrelevance axiom to assume a single universe closed under $\later$~\cite{clottmodel}. 
To solve this, CloTT has a family of universes indexed by finite sets of clocks. 
Figure~\ref{fig:universes} shows some basic rules for these. Judgemental equality rules state that codes for type operations
indeed encode those type operation, as
well as commutativity rules for universe inclusions $\univin{\Delta}{\Delta'}$ and codes for type operations. For example
\begin{align*}
% \univin{\Delta}{\Delta'}(\univin{\Delta'}{\Delta''}(t)) & = \univin{\Delta}{\Delta''}(t) & 
 \elems{\Delta'}(\univin{\Delta}{\Delta'}(t)) & = \elems{\Delta}(t) & 
 \elems{\Delta}(\latbindcode\tickA\kappa A) & = \latbind\tickA\kappa (\elems\Delta(A))
\end{align*}
Note that
a universe $\univ\Delta$ is only closed under $\later^\kappa$ for the $\kappa$ in $\Delta$. The notation $\Delta, \kappa$ is
shorthand for the union of $\Delta$ and $\{\kappa\}$. 

Unlike previous presentations of CloTT, we include a family of universes of propositions $\prop\Delta$, also indexed by clock
contexts $\Delta$. All rules of Figure~\ref{fig:universes} have corresponding rules for universes of propositions. Moreover,
these universes are closed under basic operations of predicate logic such as
%\[
\begin{mathpar}
\inferrule*{\hastype\Gamma t{\univ\Delta} \\ \hastype\Gamma\phi{\elems\Delta(t) \to \prop\Delta} 
}{\hastype\Gamma{\exists (x:\elems\Delta(t)) . \phi(x)}{\prop\Delta}
} \and
\inferrule*{ \hastype\Gamma t{\univ\Delta} \\ \hastype\Gamma {u,s}{\elems\Delta(t)}
}{ \hastype\Gamma{\peqcode us}{\prop\Delta}
}
\end{mathpar}
%\]
and a similar rule for universal quantification. 

%
%\todo{What do I need to discuss? Ticks? Perhaps if I want to claim that these are modelled also for higher ordinals later. 
%Certainly universes, also those of propositions. Axioms: Extensionality, tick irrelevance?, clock irrelevance, which is an axiom here. 
%Force}

\section{Encoding Coinductive Types}
\label{sec:encoding:coind}

One benefit of universes is that they allow for encoding guarded recursive types simply as fixed points for
endomaps on universes. Suppose, for example, that $A: \univ{\Delta,\kappa}$. Defining $\syncode{\grD} (A)$
as the fixed point of the map 
\[
\lambda (x : \later^\kappa\univ{\Delta, \kappa}) . A \syncode{+} \latercode x
\]
then $\grD (A) \defeq \elems{\Delta,\kappa}(\syncode{\grD} (A))$ satisfies $\grD(A) = A + \later^\kappa \grD(A)$. Note that 
this is an actual judgemental equality of types. In Clocked Cubical Type Theory, fixed points only unfold up to propositional
equality, and the above type equation is only an equivalence of types. The type constructor $\grD$ is often referred 
to as the \emph{guarded delay monad}. % REF. 

Guarded recursive types can be used as a stepping stone towards defining coinductive types. For example, 
the more standard coinductive delay monad $\ciD$ can be encoded as the map that maps $A : \univ\Delta$ to 
$\forallcode{}\kappa .\syncode{\grD}(A)$ where $\kappa$ is a fresh clock. This gives the coinductive solution
to $\ciD(A) \iso A + \ciD(A)$. Using guarded recursion has the benefit that one can program with $\ciD$ using the guarded
fixed point operator, which allows productivity of coinductive definitions to be captured in types. 

To describe the general setting of the encoding of coinductive types, note first that for any type $I$, the type $I \to \univ{\Delta}$ describes
the objects of a category in the naive sense: Morphisms from $X$ to $Y$ are simply terms of the type $\Pi(i:I). X(i) \to Y(i)$. Here we
have left the `elements of' operation $\elems\Delta$ implicit as we shall do below also for operations of universe inclusions
to avoid clutter. If $\Delta \subset \Delta'$, the universe inclusion $\univin\Delta{\Delta'}$ induces a functor from $I \to \univ{\Delta}$
to $I \to \univ{\Delta'}$. In the opposite direction, universal quantification over clocks defines a functor 
$(\forall\kappa. (I \to \univ{\Delta, \kappa})) \to (I \to \univ{\Delta})$ (assuming $\kappa$ fresh for $I$) by
\[
 (\forall\kappa . X)(i) \defeq \forallcode{}\kappa. \capp{X}(i)
\]
where $\forall\kappa. (I \to \univ{\Delta, \kappa})$ is considered a category in the naive sense. In the next definition
we write $\Func \ccat \dcat$ for the type of functors between categories $\ccat$ and $\dcat$. 
\begin{definition}
 A functor $F : (I \to \univ{\Delta}) \to (J\to \univ{\Delta})$ \emph{commutes with clock quantification} 
 if there exists a functor $F' : \forall\kappa . \Func{I \to \univ{\Delta, \kappa}}{J\to \univ{\Delta, \kappa}}$ extending $F$ 
 in the sense that
 \[
\begin{tikzcd}
 (I \to \univ{\Delta}) \ar[r, "F"] \ar{d} & (J\to \univ{\Delta}) \ar[d] \\
 (I\to \univ{\Delta,\kappa}) \ar[r, "F'"] & (J\to \univ{\Delta, \kappa})
\end{tikzcd}
\]
such that the canonical morphism 
 \[
  F(\forall\kappa. X) \to \forall\kappa . \capp{F'}(X)
 \]
 is an isomorphism in $J \to \univ{\Delta}$ for all $X$ in 
 $\forall\kappa . (I \to \univ{\Delta, \kappa}))$. 
\end{definition}

Suppose, for example, that $A : \univ\Delta$. Then $F(X) \defeq A + X$ defines a functor $\univ\Delta \to \univ\Delta$ (omitting $I=1$),
which extends to a functor $F' : \forall\kappa . (\univ{\Delta,\kappa} \to \univ{\Delta,\kappa})$ because the universe inclusions 
commute with binary sums. To prove that $F$ commutes with clock quantification, one needs universal 
quantification to distribute over sums 
($\forall\kappa . (B+C) \iso \forall\kappa. B + \forall\kappa .C$) and that $A$ is clock irrelevant which is an axiom (\ref{eq:cirr}).
In the following we shall use the same symbol for the functor $F$ and its extension $F'$. 

\begin{theorem}[\cite{atkey13icfp,Mogelberg14}] \label{thm:coinductive:encoding}
 Let $F : (I \to \univ\Delta) \to (I \to \univ\Delta)$ be a functor commuting with universal quantification over clocks. 
 There is a unique object $\mu^\kappa X. F(\later^\kappa X)$ in $(I \to \univ{\Delta, \kappa})$ satisfying 
 $F(\later^\kappa(\mu^\kappa X. F(\later^\kappa X))) \iso \mu^\kappa X. F(\later^\kappa X)$. The 
 object $\forall\kappa . \mu^\kappa X. F(\later^\kappa X)$ is a final coalgebra for $F$. 
\end{theorem}

The encoding was first proved with respect to a denotational model by Atkey and McBride~\cite{atkey13icfp}, and later internalised in type theory 
by the author~\cite{Mogelberg14}. Although the setting above is slightly different from the original one, due to the indexing of universes, the proof carries over
verbatim from~\cite{Mogelberg14}. To get a feeling for the proof, note how the structure isomorphism of the final coalgebra can be easily 
constructed:
\begin{align*}
\forall\kappa . \mu^\kappa X. F(\later^\kappa X) & \iso \forall\kappa . F(\later^\kappa(\mu^\kappa X. F(\later^\kappa X))) \\
 & \iso F(\forall\kappa . \later^\kappa(\mu^\kappa X. F(\later^\kappa X))) \\
 & \iso F(\forall\kappa . \mu^\kappa X. F(\later^\kappa X)) 
\end{align*}
A version of Theorem~\ref{thm:coinductive:encoding} can be proved also for the universes of propositions $\prop\Delta$. This allows for 
encoding of coinductive predicates. 

For Theorem~\ref{thm:coinductive:encoding} to be useful, one needs a large collection of functors commuting with clock quantification. Many
functors have this property just by virtue of being defined by type constructors that commute with clock quantification. For example, 
$\forall\kappa$ distributes over $\Sigma$-types just like $\Pi$ does, and commutes with $\Pi$, so polynomial functors commute with clock quantification: 
\begin{align*}
 \forall\kappa . (\Sigma (a: A) . \Pi(b : B(a) . X)) & \iso  \Sigma (a: \forall\kappa  . A) . \forall\kappa . \Pi(b : B(\capp a) . X) \\
 & \iso  \Sigma (a: A) . \forall\kappa . \Pi(b : B(a) . X) \\
  & \iso \Sigma (a: A) . \Pi(b : B(a) . \forall\kappa . X)
\end{align*}
using clock irrelevance of $A$ in the second line. In the model, universal quantification over clocks distributes over binary sums of types, 
but this is not directly provable in the type theory unless one encodes sum as a $\Sigma$-type indexed by booleans and uses clock-irrelevance
for booleans. In this paper we will just consider this an axiom, and so it proves correctness of the encoding of the coinductive
delay monad. The model also verifies that $\forall\kappa$ commutes with $\later^{\kappa'}$, and this can be proved internally 
if the rule for tick application is extended as in~\cite{CubicalCloTT}. In this paper we will look at conditions for $\forall\kappa$ commuting with 
quotient inductive types, in particular monads induced by algebraic theories. 

On the propositional side, one can prove that $\forall\kappa$ commutes with conjunction, as well
as universal quantification over clocks as well as ordinary variables. 
Similarly to the case of types, the model also verifies that $\forall\kappa$ commutes with disjunction and $\later^{\kappa'}$. 
Existential quantification is more difficult. In Clocked Cubical Type Theory, existential quantification can be encoded using $\Sigma$ 
and propositional truncation, and both of these can be proved to commute with clock quantification. In the extensional model
clock quantification does not commute with propositional truncation (see Example~\ref{ex:prop:trunc} below). 
We will see, however, that by varying the indexing ordinal of 
the model, it is safe to assume commutativity with quantification over sets of small cardinality. E.g., by indexing over $\omega_1$ 
(the first uncountable ordinal), clock quantification commutes with existential quantification over countable set. 
We first look at what is needed for applications. 

%Clock quantification 
%does not in general commute with existential quantification in the extensional model. 
%
%This paper is concerned with extending the applicability of this encoding to larger classes of functors.

\section{Applications}
\label{sec:applications}

This section recalls the applications of Clocked Cubical Type Theory to coalgebra~\cite{CubicalCloTT},  
semantics of programming languages with probabilistic choice~\cite{POPL25} and to 
reasoning about weak bisimilarity~\cite{FewExtraPages}. The aim is to determine the precise requirements 
needed to interpret these in a model of Clocked Type Theory. 
%
%guarded type theory to determine which constructions need to commute with 
%clock quantification. We focus on coalgebra and the example applications of Clocked Cubical Type Theory in \cite{POPL25} and~\cite{FewExtraPages}. 

\subsection{Coalgebra}

Coalgebras in general describe systems, and final coalgebras behaviours of such systems. Often 
the systems of interest are non-deterministic or probabilistic, and can be described as coalgebras for functors involving powersets or distributions. 
For example, coalgebras for the functor $F(X) = \Pfin(A\times X)$, where $\Pfin$ is the finite powerset functor, are finitely branching 
labelled transition systems with labels in $A$. Similarly, coalgebras for $F(X) = A \times \Dfin(X)$, where $\Dfin$ is the finite distribution functor, 
are Markov-processes with labels in $A$. Since clock quantification commutes with products, Theorem~\ref{thm:coinductive:encoding}
will give encodings of final coalgebras for these functors, if $\Pfin$ and $\Dfin$ commute with quantification over clocks. Note that $\Pfin$ and
$\Dfin$ are both generated by algebraic theories. 
For example, $\Pfin(X)$ is the free join semilattice on $X$, so generated by $\emptyset$, finite union,
and equations for associativity, commutativity, idempotency and a unit axiom for $\emptyset$. 
(See Section~\ref{sec:syntactic:condition} for details.) The natural way to program with these in type theory is 
therefore as quotient inductive types, using the associated recursion and induction principles. As an example, we recall the 
induction principle for $\Pfin(X)$ just in the case of predicates: 
Given $\phi : \Pfin(X) \to \prop\Delta$, in order to prove $\forall A : \Pfin(X). \phi(A)$, it suffices to prove
$\phi(\emptyset)$, $\phi(\{x\})$ for all $x : X$, and that $\phi(A)$ and $\phi(B)$ implies $\phi(A \cup B)$. 

Given a coalgebra $\xi : X \to F(X)$ one would also like to define notions of similarity and bisimilarity as coinductive predicates. For example, 
in the case of $F(X) = \Pfin(A\times X)$, $R$ is a simulation relation if whenever $R(x,x')$ holds, also the following holds
\[
 \forall a,y . (a,y) \in \xi(x) \to \exists y'. (a,y') \in \xi(x') \wedge R(y,y')
\]
and simulation is the largest simulation relation. Here, the predicate $(a,y) \in \xi(x)$ is defined by $\Pfin$-recursion. 
Supposing, for example, that $A, X : \univ\emptyset$, simulation should be 
the final coalgebra for the functor mapping $R : X \times X \to \Prop_\emptyset$ to 
\[
 G(R)(x,x') \defeq \forall a,y . (a,y) \in \xi(x) \to \exists y'. (a,y') \in \xi(x') \wedge R(y,y')
\]
In order to define simulation, one first defines $\similar^\kappa$ by guarded recursion to satisfy
\[
  x \similar^\kappa x' \defeq \forall a,y . (a,y) \in \xi(x) \to \exists y'. (a,y') \in \xi(x') \wedge \later^\kappa (y \similar^\kappa y')
\]
then define $x\similar x'$ as $\forall\kappa . x \similar^\kappa x'$. This definition will work if $G$ commutes with clock quantification. 
The crucial step in showing that, is to show that existential quantification and universal quantification over clocks commute. As we shall see, 
this is not a property that holds generally in the model. However, in this particular case, the natural thing to do is to redefine the predicate 
$\exists y'. (a,y') \in \xi(x') \wedge R(y, y')$ by $\Pfin$-recursion over $\xi(x')$:
\begin{itemize}
\item If $A = \emptyset$ then $\exists y'. (a,y') \in A \wedge R(y, y')$ is false
\item If $A = \{(a',y')\}$ then $\exists y'. (a,y') \in A \wedge R(y, y')$ is defined to be $a = a' \wedge R(y, y')$
\item If $A = B \cup C$ then $\exists y'. (a,y') \in A \wedge R(y, y')$ is defined as 
\[(\exists y'. (a,y') \in B \wedge  R(y, y')) \vee (\exists y'. (a,y') \in C \wedge R(y, y'))\]
\end{itemize}
Commutativity of the two quantifications can then be proved by $\Pfin$-induction if just $\forall\kappa. (\phi \vee \psi)$ 
is equivalent to $(\forall\kappa . \phi) \vee (\forall\kappa . \psi)$
and $\forall\kappa . \false$ is false. Both of these hold in the presheaf model of Clocked Type Theory independently of the indexing ordinal. 
%If $\Pfin$
%is replaced by a countable powerset functor, the case would require $\forall \kappa$ to commute with countable disjunction.  

%\todo{Probabilistic similarity?? Similarity is equality.}

\subsection{Modelling Programming Languages with Probabilistic Choice}

Stassen et al.~\cite{POPL25} construct a model of Probabilistic FPC, a programming language with recursive types and finite probabilistic choice, in 
Clocked Cubical Type Theory. One of the basic ingredients in this interpretation is a monad modelling the combination of recursion and 
probabilistic choice. In its guarded form, this monad maps a type $A : \univ\emptyset$ to the guarded recursive type 
$\grD A$ satisfying $\grD A \iso \Dfin(A + \later^\kappa(\grD A))$ where $\Dfin$ is the finite distribution monad. The coinductive version of 
this monad satisfies $\Dforall A \iso \Dfin(A + \Dforall A)$. Elements of 
this monad are distributions of values of type $A$ and computations that can run for one more step. 
As before, the correctness of the definition $\Dforall A \defeq \forall\kappa. \grD A$
requires $\Dfin$ commutes with clock quantification. 

The monads $\grD$ and $\ciD$ are used for both denotational and operational semantics. In order to prove an adequacy theorem, 
the paper constructs a  relation between syntax and semantics using guarded recursion. 
One crucial ingredient in this definition is a notion of lifting of a relation on values 
to one on computations. In terms of the universes of this paper, one can express this lifting as mapping a relation 
$R : A \to B \to \prop{\Delta, \kappa}$ where $A : \univ{\Delta,\kappa}$ and $B : \univ\Delta$ to a relation 
$\overline R^\kappa : \grD A \to \Dforall B \to \prop{\Delta,\kappa}$. Note the asymmetry: $A$ and $B$ live in different universes, and 
different monads are applied to them. The definition of $\overline R^\kappa$ is by case analysis of the first argument: An element in $\grD A$ is either
a distribution just of values of type $A$, one just of delayed computations, or a convex combination of the two: $\mu_1 \oplus_p \mu_2$ where 
$\mu_1 : \Dfin(A)$ and $\mu_2 : \Dfin(\later^\kappa (\grD A))$. (Here and below we leave the inclusions of $\Dfin(A)$ and $\Dfin(\later^\kappa (\grD A))$
into $\grD A$ implicit for readability.)
Which of the three cases an input is in is decidable and, in the latter case 
$\mu_1$ and $\mu_2$ and $p$ can be computed from the input. In the latter case the definition reads
\[
 (\mu_1 \oplus_p \mu_2) \overline R^\kappa \nu \defeq \exists \nu_1 : \Dfin(B), \nu_2 : \Dfin(\Dforall B) . 
 (\nu \leadsto^\approx \nu_1 \oplus_p \nu_2) \wedge \mu_1 \overline R^\kappa \nu_1 
 \wedge \mu_2 \overline R^\kappa \nu_2 .
\]
where $\leadsto^\approx$ is a form of reduction relation on $\Dforall$. 

One important result of the paper (Lemma~7.3) states that if $A : \univ\Delta$, 
the relation $\overline R : \Dforall A \to \Dforall B \to \prop\Delta$ defined as $\overline R (\mu, \nu) \defeq \forall\kappa . \overline R^\kappa (\capp \mu, \nu)$
is equivalent to a coinductive definition of $\overline R$ which is similar to the definition of $\overline R^\kappa$. The proof uses 
that this coinductive definition can be defined via a guarded recursive definition, and this in turn relies on the existential quantifications of 
$\nu_1$ and $\nu_2$ to commute with quantification over clocks. Unlike the existential quantifications used in the previous paragraph, 
these quantifications cannot be rewritten to inductions over a quotient inductive type. In the models presented in this paper, existential quantification 
will only commute with clock quantification if the existential quantification is of bounded size. The result as stated in this generality will therefore not be true 
in these models. However, Stassen et al. only apply this theorem to the very specific case of $\nu : \Dforall 1$, and the quantifications are therefore over 
$\nu_1 : \Dfin(1)$ and $\nu_2 : \Dfin(\Dforall 1)$, so of bounded size. % The set $\Dforall 1$ is uncountable. 

\subsection{Weak Bisimilarity}

The coinductive delay monad satisfying $\ciD(A) \iso A + \ciD(A)$, 
is an intensional model of recursion because it counts computation steps. Often one wants to reason up to weak bisimilarity. 
However, working with weak bisimilarity is quite subtle. For example, proving that the quotient of the coinductive delay monad by weak bisimilarity
is a monad seems to require countable choice~\cite{quotientingDelay}. 
Rather than working with the quotient, one therefore often works in a category of setoids. 
For example, Møgelberg and Zwart~\cite{FewExtraPages} prove that if $T$ is a free model monad for an algebraic theory with no drop equations, 
then there is a distributive law of T over the coinductive delay monad quotiented as a setoid. 

Recall that coinductive delay monad %$\ciD A$ satisfying $\ciD(A) \iso A + \ciD(A)$ 
can be defined as $\ciD (A) \defeq \forall\kappa . \grD(A)$ where 
$\grD(A) \iso A + \later^\kappa \grD (A)$. Given a setoid $(X,R)$, we would like to define the setoid $(\ciD(X), \sim_R)$
capturing weak bisimilarity. 
Writing
\begin{align*}
 \now & : A \to \ciD(A) & \step : \ciD(A) \to \ciD(A) \\
 \now^\kappa & : A \to \grD(A) & \step^\kappa : \later^\kappa\grD(A) \to \grD(A) 
\end{align*}
for the left and right inclusions into the sum, the relation $\sim_R$ should be defined by coinduction and cases:
\begin{align*}
 \now(a) \sim_R y & \defeq \exists (n : \N, y' : X) . y = \step^n(\now(y')) \wedge R(a, y') \\
 x \sim_R \now (a) & \defeq \exists (n : \N, x' : X) . x = \step^n(\now(x')) \wedge R(x', a) \\
 \step (x) \sim_R \step (y) & \defeq x\sim_R y
\end{align*}
This coinductive predicate can be encoded using guarded recursion, because the function mapping a relation $\sim_R$ to $F(\sim_R)$ 
defined as above commutes with clock quantification. Indeed, $x\sim_R y$ does not appear in the two first clauses,
and so in these cases the required condition follows from the axiom of clock irrelevance. In the last clause the case is immediate. 

However, in order to use guarded recursion for programming with weak bisimilarity, Møgelberg and Zwart~\cite{FewExtraPages} 
use a different definition: They define 
\[
  x \sim_R y \defeq \forall\kappa . \capp x \sim_R^\kappa \capp y
\]
where $\sim_R^\kappa : \grD X \to \grD X \to \prop{\Delta,\kappa}$ is defined also for
$X : \univ{\Delta,\kappa}$ and $R : X \to X \to \prop{\Delta,\kappa}$ as the unique relation satisfying
\begin{align*}
 \now^\kappa(a) \sim_R^\kappa y & \defeq \exists (n : \N, y' : X) . y = (\step^\kappa \circ \nextop)^n(\now^\kappa(y')) \wedge R(a, y') \\
 x \sim_R^\kappa \now^\kappa (a) & \defeq \exists (n : \N, x' : X) . x = (\step^\kappa \circ \nextop)^n(\now^\kappa(x')) \wedge R(x', a) \\
 \step^\kappa (x) \sim_R^\kappa \step^\kappa (y) & \defeq \latbind\tickA\kappa{ (\tapp x \sim^\kappa_R \tapp y)}
\end{align*}
Is this the same as the usual coinductive definition? When comparing the two we can assume $X : \univ\Delta$ and 
$R : X \to X \to \prop\Delta$. The critical cases are the first two, so suppose $x$  is $\now(a)$. Then $\forall\kappa . \capp x \sim_R^\kappa \capp y$
reduces to 
\[
\forall\kappa. \exists (n : \N, y' : X) . \capp y = (\step^\kappa \circ \nextop)^n(\now^\kappa(y')) \wedge R(a, y')
\]
We would like to commute $\forall\kappa$ past the two existential quantification, after which the equivalence follows by the relationship
between $\step$ and $\step^\kappa$ and the clock irrelevance axiom applied to $R(a,y')$. The first quantification is over a countable 
set, so will commute with $\forall\kappa$ if the indexing in the model is over a sufficiently high ordinal. The other one is more difficult
because we want to define a generally applicable monad, and so should not put restrictions on the cardinality of $X$. 
On the other hand, the $y'$ is essentially unique because 
\[
(\step^\kappa \circ \nextop)^n(\now^\kappa(y')) = (\step^\kappa \circ \nextop)^n(\now^\kappa(y''))
\]
implies $\later^n(y' = y'')$. It therefore suffices to show that clock quantification commutes with existential 
quantifications that are unique in this sense. 
%Alternatively, one could change the definition to
%\[
% \exists (n : \N, y' : (\later^\kappa)^nX) . y = (\step^\kappa)^n(\now^\kappa(y')) \wedge R(a, y')
%\]

%\todo{Each $\grD(A)$ is same size as $A$ unless $A$ countable. Would like to avoid such assumptions to define a monad}

%\todo{Discuss precisely which properties I need: Commutativity of forall and existential quant over small sets and some form of unique existential quantification. 
%What do I do with existential quantification over elements of elements of a finite powerset?}

\subsection{Summary Requirements}
\label{sec:requirements}

Having analysed the applications in detail, we now summarise the requirements needed for universal 
quantification over clocks in order to express the applications above. They are
\begin{itemize}
\item Universal quantification needs to commute with certain quotient-inductive type constructors such as finite powersets and finite distributions
\item Universal quantification needs to commute with  bounded existential quantification, so over sets that are below a certain 
cardinality
\item Universal quantification needs to commute with existential quantification that is essentially unique in the sense that if $X : \univ\Delta$ and
$\phi: X \to \prop{\Delta, \kappa}$ is such that $\phi(x) \wedge \phi(y)$ implies $(\later^\kappa)^n(x=y)$, then 
$\forall\kappa . \exists x:X. \phi(x)$ implies $\exists (x:X). \forall\kappa . \phi(x)$. 
\end{itemize}

\section{The Presheaf Model}
\label{sec:model}

The section recalls the model of multi-clocked guarded recursion described by Mannaa et al.~\cite{clottmodel}. The model is 
indexed by a limit ordinal $\indexord$, but Mannaa et al. only describe the specific case of $\indexord = \omega$. All constructions generalise 
directly, however\footnote{The abstract description of the left adjoint to $\tlater$ does not immediately generalise, but the concrete version,
which is used throughout that paper, does.}. 
The model uses a presheaf category over a category $\catT$ of time objects defined as pairs $\timeobj\FSA\vartheta$ where
$\FSA$ is a finite subset of some fixed countably infinite set of clock names and $\vartheta : \FSA \to \indexord$. A morphism from $\timeobj\FSA\vartheta$ 
to $\timeobj\FSB{\vartheta'}$ is a map $\sigma : \FSA \to \FSB$ such that $\vartheta' \circ \sigma \leq \vartheta$ in the pointwise order. 

Let $\grtotal$ be the category $\PSh(\catT) \defeq \Set^{\catT}$ of \emph{covariant} presheaves on $\catT$. We write $\sigma \cdot x$ for $X(\sigma)(x)$ when 
$X$ is a presheaf, $x \in X\timeobj\FSA\vartheta$ and $\sigma :\timeobj\FSA\vartheta \to \timeobj\FSB{\vartheta'}$. There is an object of clocks 
$\clk$ in $\grtotal$ defined as $\clk\timeobj\FSA\vartheta = \FSA$. This is used to model assumptions of the form $\kappa : \clocktype$
in contexts. 

The operation $\later$ can be described semantically using an endofunctor $\tlater$ on the slice category $\slice \grtotal\clk$. To describe this, 
first recall that any slice category of a presheaf category is equivalent to the category of presheaves over the elements of the 
slicing object. In this particular case, the category of elements $\int\clk$ has as objects triples 
$\triple\FSA\vartheta\lambda$, where $\lambda \in \FSA$ 
and morphisms from $\triple\FSA\vartheta\lambda$ to $\triple\FSB{\vartheta'}{\lambda'}$ are $\catT$-morphisms 
$\sigma : \timeobj \FSA\vartheta \to \timeobj\FSB{\vartheta'}$ such that $\sigma(\lambda) = \lambda'$. The modality $\tlater$ can then be defined 
as 
\begin{align*}
 (\tlater X)\triple \FSA\vartheta\lambda & = \lim_{\alpha < \vartheta(\lambda)} X\triple\FSA{\vartheta[\lambda \mapsto \alpha]}{\lambda}
\end{align*}
Note how this generalises the definition of $\tlater$ in the case of the topos of trees ($\PSh(\opcat\omega)$): $(\tlater X)(0) = 1$, $(\tlater X)(n+1) = X(n)$
to higher ordinals and the multi-clocked setting. 

The model constructed by Mannaa et al. follows the standard construction of a semantic model of type theory from presheaf models: 
Semantic contexts are objects $\Gamma$ of $\grtotal$, and semantic types in semantic context $\Gamma$ are presheaves over $\int \Gamma$. 
So a semantic type assigns, to each $\timeobj\FSA\vartheta$ and $\gamma \in \Gamma\timeobj\FSA\vartheta$ a set $X\triple\FSA\vartheta\gamma$, 
and to each $\sigma : \timeobj\FSA\vartheta \to \timeobj\FSB{\vartheta'}$ a morphism $\sigma \cdot(-) : X\triple\FSA\vartheta\gamma \to 
X\triple\FSB{\vartheta'}{\sigma \cdot \gamma}$. In order to interpret the clock irrelevance axiom, the model restricts attention to 
semantic types that are clock irrelevant in the following sense. 

\begin{definition}[\cite{GDTTmodel}] \label{def:invariant:under:clock:intro}
 Let $\Gamma$ be an object of $\grtotal$. A presheaf $X$ over $\int\Gamma$ is said to be \emph{invariant under clock introduction}
 if, for any $\timeobj\FSA\vartheta$, $\lambda\notin \FSA$ and $\alpha$, the map 
 \[
  \iota \cdot(-) :  X\triple\FSA\vartheta\gamma \to X\triple{(\FSA,\lambda)}{\vartheta[\lambda \mapsto \alpha]}{\iota \cdot \gamma}
 \]
 is an isomorphism, where $\FSA, \lambda$ is shorthand for $\FSA \cup \{\lambda\}$ and 
 $\iota : \FSA \to (\FSA,\lambda)$ is the inclusion. 
\end{definition}
For example, if $\Gamma$ is the terminal object $1$, a presheaf $X$ in $\PSh(\int 1) \iso \grtotal$ is invariant under clock introduction if and 
only if each $\sigma\cdot(-) : X\timeobj\emptyset{!} \to X\timeobj\FSA{\vartheta}$ is an isomorphism, and so $X$ is a constant presheaf. 
Similarly, the subcategory of presheaves over $\int\clk$ invariant under clock introduction is equivalent to the topos of trees $\Set^{\opcat\omega}$. 
On the other hand, $\clk$ is an example of a presheaf that is not invariant under clock introduction. 

The collection of presheaves invariant under clock introduction is closed under standard type formers such as $\Sigma$, $\Pi$ and universal 
quantification over clocks~\cite{GDTTmodel}. Universal quantification over clocks in fact is just modelled as a $\Pi$-type indexed by $\clk$. 
We recall~\cite{GDTTmodel} an alternative, direct definition of this which we will need in the following sections. Suppose $\Gamma$ is an object of $\grtotal$
and $A$ is a (semantic) dependent type over $\Gamma\times \clk$. This means that for any $\timeobj\FSA\vartheta$, $\gamma\in \Gamma\timeobj\FSA\vartheta$
and $\lambda\in \FSA$, $A\quadruple\FSA\vartheta\gamma\lambda$ is a set. Suppose we have a function that assigns, to each $\FSA$ an element
$\lambda_\FSA \notin \FSA$, and define $\forall(A)$ as a dependent type over $\Gamma$ by 
\[
 \forall(A)\triple\FSA\vartheta\gamma \defeq \lim_{\alpha<\indexord} A\quadruple{(\FSA,\lambda_\FSA)}{\vartheta[\lambda_\FSA \mapsto \alpha]}{\iota\cdot\gamma}{\lambda_\FSA}
\] 
where $\iota : \FSA \to (\FSA, \lambda_\FSA)$ is the inclusion. The axiom (\ref{eq:force}) can then be proved as follows 
\begin{align*}
 \forall(\tlater A)\triple\FSA\lambda\gamma & \defeq \lim_{\alpha<\indexord} \lim_{\beta<\alpha} A\quadruple{(\FSA,\lambda_\FSA)}{\vartheta[\lambda_\FSA \mapsto \beta]}{\iota\cdot\gamma}{\lambda_\FSA} \\
 & \iso\lim_{\alpha<\indexord} A\quadruple{(\FSA,\lambda_\FSA)}{\vartheta[\lambda_\FSA \mapsto \alpha]}{\iota\cdot\gamma}{\lambda_\FSA} \\
 & = \forall(A)\triple\FSA\lambda\gamma
\end{align*}
which works only because $\indexord$ is assumed to be a limit ordinal. Note also that this definition makes it clear that 
clock quantification distributes over sums. 

\begin{terminology}
 In the following, given an object $\Gamma$ of $\grtotal$, we shall refer to a presheaf $X$ over $\int\Gamma$ which is invariant under clock introduction
 as a \emph{family} over $\Gamma$. The sets $X\triple\FSA\vartheta\gamma$ are referred to as the \emph{fibres} of $X$, and if all fibres of $X$ are
 subsingletons, we say that $X$ is a predicate on $\Gamma$. 
\end{terminology}

%\subsection{Universal Quantification over Clocks}
%
%\todo{Explicit construction. Closure under exists. }

\subsection{Universes}

If $\Delta$ is a clock context (so a finite set of clock variables), form the presheaf $\clk^\Delta$ in $\grtotal$ as the $\Delta$-fold product
of $\clk$ by itself. Write $\gr{\Delta}$ for $\PSh(\int \clk^\Delta)$. Assuming a set theoretic universe of small sets, 
there is an object $\Usem\Delta$ in $\gr\Delta$ and a family $\Elsem\Delta$ over it such that 
\begin{itemize}
\item $\Usem\Delta$ is 
%and $\Elsem\Delta$ are 
invariant under clock introduction
\item For any $\Gamma$ in $\gr\Delta$ and $A$ a family over $\Gamma$ with small fibres,
there is a unique map $\code A : \Gamma \to \Usem\Delta$ such that $A = \Elsem\Delta[\code A]$. 
\item This universe is closed under $\Pi$- and $\Sigma$-types, universal quantification over clocks, $\later^\kappa$ for 
$\kappa \in \Delta$. 
\end{itemize}
The construction uses a variant of the Hofmann-Streicher construction of universes in presheaf categories~\cite{Hofmann-Streicher:lifting}: 
An element of $\Usem\Delta\triple\FSA\vartheta {f: \Delta \to \FSA}$, is a family over $y(f[\Delta], \restrict{\vartheta}{f[\Delta]})$ in 
$\grtotal$ with small fibres, where $y$ is the yoneda embedding, and $f[\Delta]\subseteq \FSA$ is the image of $f$. 
Concretely, this means that an element is a family of small sets 
$A_\sigma$ indexed by the diagram whose components
are maps $\sigma : (f[\Delta], \restrict{\vartheta}{f[\Delta]}) \to \timeobj\FSB{\vartheta'}$, which is invariant under clocks in the sense
that $\iota\cdot(-) : A_{\sigma} \to A_{\iota\circ\sigma}$ is an isomorphism for all $\iota$ as in Definition~\ref{def:invariant:under:clock:intro}. 
For example, $\univ\emptyset$ is an object in $\gr\emptyset \iso \grtotal$, and this is 
the constant presheaf of small $\grtotal$ objects that are invariant under clock introduction. 
The universes of Clocked Type Theory can be modelled using these universes by appropriate reindexing. We refer to 
Mannaa et al.~\cite{clottmodel} for further details. 
For this paper we also need a family of universes of propositions.

\begin{proposition}
 There exists an object $\Propsem\Delta$ of $\gr\Delta$ and a predicate $\PElsem\Delta$ on it such that 
\begin{enumerate}
%\item Both $\Propsem\Delta$ and $\PElsem\Delta$ are invariant under clock introduction
\item $\Propsem\Delta$ is invariant under clock introduction
%\item The fibres of $\PElsem\Delta$ are all subsingletons 
\item\label{item:prop:classifies} For any $\Gamma$ in $\gr\Delta$ and $A$ predicate on $\Gamma$, 
there is a unique map $\code A : \Gamma \to \Propsem\Delta$ such that $A \iso \PElsem\Delta[\code A]$. 
\item This universe is closed under $\forall$- and $\exists$-quantification over objects invariant under clock introduction, 
finite conjunction and disjunction, universal quantification over clocks, and $\later^\kappa$ for $\kappa \in \Delta$. 
\end{enumerate} 
\end{proposition}

\begin{proof}
Let $\Propsem\Delta\triple\FSA\vartheta {f: \Delta \to \FSA}$ be the subset of $\Usem\Delta\triple\FSA\vartheta {f: \Delta \to \FSA}$ 
consisting of families where each $A_\sigma$ is a subset of $1= \{\star\}$. The proof of almost all the properties above
then carry over verbatim from the case of $\Usem\Delta$, except for closure under existential quantification. 
For that, suppose $\Gamma$ is an object of $\gr\Delta$, $A$ a family over $\Gamma$, 
$\code B : \compr\Gamma A \to \Propsem\Delta$, and $B = \PElsem\Delta[\code B]$. Here 
\[
  \compr\Gamma A\timeobj\FSA\vartheta = \{ (\gamma, a) \mid \gamma \in \Gamma\timeobj\FSA\vartheta \wedge a \in A\triple\FSA\vartheta\gamma\}
\]
We must show that the presheaf $\exists(A,B)$ over $\int\Gamma$ defined by
\[
  \exists(A,B) \triple\FSA\vartheta\gamma \defeq \cup_{a \in A \triple\FSA\vartheta\gamma} B\quadruple\FSA\vartheta\gamma a
\]
has a code $\code{\exists(A,B)} : \Gamma \to \Propsem\Delta$. By item~\ref{item:prop:classifies} of the proposition, it suffices to
show that $\exists(A,B)$ is invariant under clock introduction, and so a predicate. So let 
$\iota : \timeobj\FSA\vartheta \to \timeobj{(\FSA,\lambda)}{\vartheta[\lambda \mapsto \alpha]}$
be given by the inclusion of $\FSA$ into $(\FSA, \lambda)$. Then 
\begin{align*}
 \exists(A,B) \triple{(\FSA,\lambda)}{\vartheta[\lambda\mapsto \alpha]}{\iota\cdot\gamma} 
 & \defeq \cup_{a \in A \triple{(\FSA,\lambda)}{\vartheta[\lambda \mapsto \alpha]}{\iota\cdot\gamma}} B\quadruple{(\FSA,\lambda)}{\vartheta[\lambda \mapsto \alpha]}{\iota\cdot\gamma} a \\
 & = \cup_{a \in A \triple{\FSA}{\vartheta}{\gamma}} B\quadruple{(\FSA,\lambda)}{\vartheta[\lambda \mapsto \alpha]}{\iota\cdot\gamma} {\iota \cdot a} \\
 & = \cup_{a \in A \triple{\FSA}{\vartheta}{\gamma}} B\quadruple{\FSA}{\vartheta}{\gamma} {a} \\
 & = \exists(A,B)\triple\FSA\vartheta\gamma
\end{align*}
using invariance under clock introduction for $A$ and $B$.
%
%By varying the construction using the universe of subsingletons one obtains a 
%family of universes %$\Propsem\Delta$ 
%of propositions. 
%  
\end{proof}

\subsection{Conservativity Over $\Set$}

One important consequence of the requirement that families be invariant under clock introduction 
is that the fragment of the type theory that does not mention clocks, ticks or guarded 
recursion is interpreted using constant presheaves only, and so the interpretation of this fragment agrees with that of the standard 
model in $\Set$. In particular, any proof of such a statement in the model based on $\grtotal$ gives rise to a proof
in the model based on $\Set$. 

\begin{theorem} \label{thm:conservative:model}
 Let $\istype\Gamma A$ be a wellformed statement in Clocked Type Theory 
 that can also be expressed in the internal type theory of $\Set$. If there is a
 term $t$ such that $\hastype\Gamma tA$, then also the dependent type constructed by interpreting $\istype\Gamma A$ in $\Set$ 
 can be inhabited.
\end{theorem}

Note that $\Gamma$ and $A$ can be constructed using guarded recursion as long as they can be expressed in the internal language
of $\Set$. For example, 
%recall that $\univ\emptyset$ in $\gr\emptyset \iso \grtotal$ is the constant presheaf of small $\grotal$ objects invariant
%under clock introduction. This presheaf is not isomorphic to the constant presheaf of small sets, but these presheaves are equivalent as
%groupoids. In particular, any 
any endofunctor $F$ on $\Set$ mapping small sets to small sets induces a functor 
$F: \univ\emptyset \to \univ\emptyset$ in the internal language
of our model. The type $\forall\kappa . \mu X. F(\later^\kappa X)$ is a closed type, so is interpreted as a set.
We can therefore view the statement that $\forall\kappa . \mu X. F(\later^\kappa X)$ is a final coalgebra for 
$F$ (Theorem~\ref{thm:coinductive:encoding}) as a statement in the internal language of $\Set$, and since this statement has 
a proof in Clocked Type Theory if $F$ commutes with clock quantification, then, 
by Theorem~\ref{thm:conservative:model}, it also holds in the set-theoretic interpretation. 
Of course, this statement will talk about the action of $F$ as an endofunctor on $\univ\emptyset$, which, as mentioned above
is the set of all small objects of $\grtotal$ that are invariant under clock introduction. However, the notion of functor expressed using
this universe in $\Set$ is equivalent to the one expressed using the usual universe of small sets. 
Many of the end results proved in the applications mentioned in Section~\ref{sec:applications}, are of this form, for example, 
existence or non-existence of distributive laws of monads on $\Set$ or setoids, and contextual equivalence of programs of Probabilistic 
FPC. Their proofs use guarded recursion, but by Theorem~\ref{thm:conservative:model} these proofs also prove their usual
set theoretic interpretation. 
%
%
%Of course, this set is not isomorphic to
%the set of all small sets, but since these categories are equivalent, the two notions of functors are equivalent.
%
%\todo{Need to say what such functors are. The fact that it is a functor also holds in }

\section{Algebraic data types} 
\label{sec:alg:types}

The construction of lifting a set functor F preserving smallness to an endomap on a universe $\Usem\emptyset$ extends to any of the universes
$\Usem\Delta$. For such a functor to commute with universal quantification over clocks means that the canonical map 
\[
 F(\lim_{\alpha<\indexord} A\quadruple{(\FSA,\lambda_\FSA)}{\vartheta[\lambda_\FSA \mapsto \alpha]}{\iota\cdot\gamma}{\lambda_\FSA}) \to 
 \lim_{\alpha<\indexord} F(A\quadruple{(\FSA,\lambda_\FSA)}{\vartheta[\lambda_\FSA \mapsto \alpha]}{\iota\cdot\gamma}{\lambda_\FSA}) 
\]
must be an isomorphism. Since $A$ can be anything, this means that $F$ needs to commute with $\opcat{\indexord}$-indexed limits. 
Theorem~\ref{thm:coinductive:encoding} is therefore
simply a synthetic restatement of Ad{\'a}mek's fixed point theorem~\cite{adamek1974free}: 
The type $\forall\kappa . \mu^\kappa X. F(\later^\kappa X)$ is the limit of the terminal 
sequence 
\[
  1 \from F(1) \from F^2(1) \from \dots \from F^\omega(1) \from F^{\omega +1} \from \dots
\] 
(where $F^\beta(1) = \lim_{\alpha < \beta} F^\alpha$ for limit ordinals $\beta$), and Ad{\'a}mek's theorem states 
that this is a terminal coalgebra if $F$ preserves $\opcat{\indexord}$-indexed limits. 

In this section we look at conditions for functors to commute with ordinal-indexed limits. We are particularly interested in the cases 
of monads generated by algebraic theories, as in the examples of finite powersets and finite distributions. First a few examples. 

\begin{example} \label{ex:fin:powset}
 The finite powerset monad $\Pfin$ does not commute with $\opcat\omega$-limits. In fact the limit of the $\opcat\omega$-long terminal sequence 
 \[
 1 \from \Pfin(1) \from \Pfin^2(1) \from \dots
 \]
 contains infinite sets~\cite{adamek1995greatest, Worrell05}. Worrel~\cite{Worrell05} proved that the terminal sequence converges after $\omega + \omega$ steps, but in general $\Pfin$ does not commute with 
 $\opcat{(\omega+\omega)}$-indexed limits. For a counter example, just append the terminal sequence to an $\omega$-long sequence of $1$. 
 In general, however, $\Pfin$ does commute with 
 $\opcat{\omega_1}$-limits, where $\omega_1$ is the first uncountable ordinal~\cite{adamek1995greatest}.
 %as first proved by Ad{\'a}mek and Koubek~\cite{adamek1995greatest}.
\end{example}

The next example shows that some monads generated by algebraic theories do not commute with any $\opcat{\indexord}$-limits.

\begin{example} \label{ex:prop:trunc}
 Propositional truncation is the monad $T(X) = \{\star \mid \exists x : X.\true\}$. This monad does not commute with 
 $\opcat{\indexord}$-limits for any $\indexord$. Consider, for example,
 the sequence
 \[
  X_\alpha = \{ \beta \in \indexord \mid \beta > \alpha \}
 \]
 with maps $X_\alpha \to X_\gamma$ for $\gamma \leq \alpha$ given by inclusions. The limit of this sequence is empty, so also $T(\lim X) = \emptyset$, 
 but $T(X_\alpha) = 1$, so $\lim T(X) = 1$. 
\end{example}

\subsection{A Semantic Condition}
\label{sec:semantic:condition}

We now state conditions that imply that a functor preserves $\opcat{\indexord}$-indexed limits. 
Let $\kappa$ be a regular cardinal. 
A set $Y$ is $\kappa$-small if it has cardinality strictly less than $\kappa$, and we 
write $Y\kapsubset X$ to mean that $Y$ is a $\kappa$-small subset of $X$.
Recall~\cite{adamek1994locally} that a functor $F : \Set \to \Set$ is said to be $\kappa$-accessible,
if it preserves $\kappa$-filtered colimits. In particular, for any $X$, the map 
\[
  \colim_{X'\kapsubset X} F(X') \to F(X)
\]
must be an isomorphism (in fact this statement is equivalent to being $\kappa$-accessible~\cite{adamek2019finitary}). 
In words, this means that any $x \in F(X)$ lives already in $F(X')$ for some $\kappa$-small subset $X'$ of $X$.
Say that a functor $F$ preserves pullbacks of monos, if, for any pullback diagram
\[
\begin{tikzcd}
 A \ar[r] \ar{d} \ar[rd, phantom, "\lrcorner"] & B \ar[d, rightarrowtail, "m"] \\
 C \ar[r] & D
\end{tikzcd}
\]
where $m$ is a mono, also $F$ applied to the diagram is a pullback. 
In particular, this means that $F$ preserves monos,
also the non-split ones from the empty set. 
% This can be seen by considering the pullback of $\emptyset \to X$ along itself. 
% It also means that $F$ preserves binary intersections of subobjects.

Say that a functor $F$ preserves $\kappa$-small intersections if, for all families  $(X_i\subseteq X)_{i\in I}$
indexed by $\kappa$-small sets $I$, $F(\cap_{i\in I} X_i)$ is the limit of the diagram 
\[
\begin{tikzcd}
F(X_i) \ar{dr} & \dots & \ar{dl} F(X_j) & \dots \\
& F(X)
% A \ar[r] \ar{d} & B \ar[d, rightarrowtail, "m"] \\
% C \ar[r, rightarrowtail, "n"] & D
\end{tikzcd}
\]

%
%binary intersections if, for all pairs of subsets  $(X_i\subseteq X)_{i = 1,2}$
%the below diagram is a pullback
%\[
%\begin{tikzcd}
%F(X_1 \cap X_2) \ar{r} \ar{d} \arrow[dr, phantom, "\lrcorner" , very near start, color=black] & F(X_1) \ar{d}  \\
%F(X_2) \ar{r} & F(X)
%% A \ar[r] \ar{d} & B \ar[d, rightarrowtail, "m"] \\
%% C \ar[r, rightarrowtail, "n"] & D
%\end{tikzcd}
%\]
%where all maps are $F$ applied to the set inclusions. 

% \begin{theorem} \label{thm:semantic:condition}
%  Let $F : \Set \to \Set$ be a finitary functor which preserves pullbacks of monomorphisms. 
%  If $\indexord$ is an uncountable regular ordinal, then $F$ preserves $\opcat{\indexord}$-indexed limits. 
% \end{theorem}

\begin{theorem} \label{thm:semantic:condition}
 Let $F : \Set \to \Set$ be a $\kappa$-accessible functor which preserves $\kappa$-small intersections and pullbacks of monomorphisms. 
 If $\indexord$ is a regular ordinal, and there exists an $\alpha\in \indexord$ of the same cardinality as $\kappa$, 
 then $F$ commutes with $\opcat{\indexord}$-indexed limits. 
\end{theorem}

For the proof we need the following. 

\begin{lemma}\label{lem:bij:from:some:point}
  Let $\indexord$ be a regural ordinal, and let $X : \opcat{\indexord} \to \Set$ be a diagram
  whose maps $X_\beta \to X_\alpha$ are all surjective. If, for any $\alpha\in \indexord$,
  there exists an ordinal  $\beta >\alpha$ such that $X_\beta \to X_\alpha$ is not injective,
  then, for any $\alpha \in \indexord$ there is 
  a $\beta(\alpha)$ and an injection $\alpha \to X_{\beta(\alpha)}$. 
\end{lemma}

\begin{proof}
  To construct $\beta(\alpha)$ we construct first a map
  $\gamma: \indexord \to \indexord$ by induction on the input: If $\alpha$ is a successor,
  so $\alpha = s(\alpha')$, define $\gamma(\alpha)$ to be an ordinal such that 
  $X_{\gamma(\alpha)} \to X_{\gamma(\alpha')}$ is not injective. If $\alpha$ is a 
  limit ordinal, define 
  \[
   \gamma(\alpha) = \bigvee_{\alpha'<\alpha}\gamma(\alpha')
  \]
  In this case $\gamma(\alpha)$ is in $\indexord$ because $\indexord$ is regular. 
  Using this, define $\beta(\alpha) = \gamma(s(\alpha))$.
  Next we construct injections $f_{\alpha'} : \alpha' \to X_{\beta(\alpha)}$ for 
  $\alpha' < s(\alpha)$ by induction on $\alpha'$ such that
  \begin{itemize}
    \item $f_{\alpha'}$ extends $f_{\alpha''}$ whenever $\alpha''<\alpha'$ and
    \item the composition $X_{\beta(\alpha), \gamma(\alpha')}\circ f_{\alpha'} : \alpha' \to X_{\beta(\alpha)} \to X_{\gamma(\alpha')}$
is injective for all $\alpha'$.
  \end{itemize}
  Again, we proceed by cases of $\alpha'$: If $\alpha' = s(\alpha'')$, since 
  $X_{\gamma(\alpha')} \to X_{\gamma(\alpha'')}$ is not injective, and $X_{\beta(\alpha), \gamma(\alpha'')}\circ f_{\alpha''}$ is
  injective, there must be an element $x\in X_{\gamma(\alpha')}$ outside the image of $X_{\beta(\alpha), \gamma(\alpha')}\circ f_{\alpha''}$.
  We can therefore extend $f_{\alpha''}$ to $f_{\alpha'}$ by picking $f_{\alpha'}(\alpha'')$ to be an element in $X_{\beta(\alpha)}$
  that maps to $x$ by $X_{\beta(\alpha), \gamma(\alpha')}$. When $\alpha'$ is a limit ordinal, we can simply define 
  $f_{\alpha'}(\alpha'') = f_{s(\alpha'')}(\alpha'')$. 
\end{proof}

\begin{proof}[Proof of Theorem~\ref{thm:semantic:condition}]%[Proof (Sketch)]
 First note that if $x \in F(X)$, then there exists a smallest subset $X'\kapsubset X$ such that $x \in F(X')$ (leaving the action of $F$
 on the inclusion implicit): There exists an $X'$ because $F$ is $\kappa$-accessible, and the smallest is then the intersection of all $X'\setminus\{x'\}$ such 
 that $x \in F(X'\setminus \{x'\})$. 

%  First note that if $x \in F(X)$, then there exists a smallest subset $X'\finsubset X$ such that $x \in F(X')$ (leaving the action of $F$
%  on the inclusion implicit): There exists an $X'$ because $F$ is finitary, and the smallest is then the intersection of the finitely many
%  subsets $X''$ of $X'$, such that $x \in F(X'')$. 
 
 Let $X$ be a $\opcat{\indexord}$-indexed diagram. We must show that the canonical map
 \[
  F(\lim X) \to \lim F(X)
 \]
 is an isomorphism. To show surjectivity, suppose $x \in \lim F(X)$. Let $X'_\alpha$ be the smallest subset  
 of $X_\alpha$ such that $x_\alpha \in F(X'_\alpha)$. We show that these form a subdiagram of $X$. Suppose $\beta \leq \alpha$ and 
 let $r : X_\alpha \to X_\beta$ be the restriction map. Then clearly $X'_{\beta} \subseteq r[X'_\alpha]$ by minimality of $X'_\beta$. Since $F$ 
 preserves pullbacks of monomorphisms, the diagram
\[
\begin{tikzcd}
  F(\inv r(X'_\beta)\cap X'_\alpha) \ar{r} \ar{d} & F(X'_\beta) \ar{d} \\
  F(X'_\alpha) \ar{r} &  F(r[X'_\alpha])
\end{tikzcd}
\]
is a pullback, so $x_\alpha \in F(\inv r(X'_\beta))$, and by minimality then $\inv r(X'_\beta)\cap X'_\alpha = X'_\alpha$. So 
\[
  r[X'_\alpha] = r[\inv r(X'_\beta)\cap X'_\alpha] \subseteq X'_\beta
\]
and so $X'_{\beta} = r[X'_\alpha]$. So $X'$ is a subdiagram, and each map in $X'$ is surjective. 

We next show that the maps $X'_\beta \to X'_\alpha$ must be bijections from some point on. If not, 
by Lemma~\ref{lem:bij:from:some:point}, for any $\alpha \in \indexord$ there must be an injection
of $\alpha$ into some $X'_{\beta(\alpha)}$, but since there is an $\alpha$ of
cardinality $\kappa$, this contradicts $X'_{\beta(\alpha)}$ being $\kappa$-small.
% there exists a sequence $(\alpha_n)_{n\in \N}$ such that each $X'_{\alpha_n}$ has at least $n$ elements. 
% Then $X'_{\vee \alpha}$ must be infinite, which is a contradiction. Note that $\vee \alpha< \indexord$ because 
% $\indexord$  is assumed regular and uncountable.

%
%construct a mapping 
%$\phi: \kappa \to \kappa$ and an injection $\psi : \Pi \alpha < \kappa . \alpha \to X'_{\phi(\alpha)}$ such that 
%\[
% r(\psi(\alpha)(\gamma)) = \psi(\beta)(\gamma)
%\]
%for $\gamma < \beta < \alpha$. These are constructed by
%induction on $\kappa$. First set $\phi(0) = 0$. For a successor ordinal $\alpha+$, choose $\phi(\alpha+)$ to be such that 
%$X'_{\phi(\alpha+)} \to X'_{\phi(\alpha)}$ is not surjective. Define, for $\beta < \alpha$, 
%$\psi(\alpha+)(\beta)$ to be some element in the preimage of $\psi(\alpha)(\beta)$ using $r : X'_{\phi(\alpha+)} \to X'_{\phi(\alpha)}$ 
%surjective. Define $\psi(\alpha+)(\alpha)$ to be some element not in $\psi(\alpha+)[\alpha]$, which must exist, since 
%$r : X'_{\phi(\alpha+)} \to X'_{\phi(\alpha)}$  is not injective. For $\alpha$ a limit ordinal, pick $\phi(\alpha) = \bigvee_{\beta < \alpha} \phi(\beta)$,
%and define $\psi(\alpha)(\beta)$ to be some element mapping to $\psi(\beta+)(\beta) \in X'_{\beta+}$. 
%\todo{How do I show this is injective? I think I would need that if $\gamma < \beta$ then $\psi(\alpha)(\gamma)$ maps
%to $\psi(\beta+)(\gamma)$, but I only know that it maps to $\psi(\gamma+)(\gamma)$. This requires consistent choices,
%and I dont know how to do that.}

So there must be a stage $\alpha$ such that all maps $X'_\beta \to X'_\alpha$ are isomorphisms. 
Then $F(\lim X') \iso F(X'_{\alpha}) \iso \lim F(X')$, and so $x$ is in the image of the comparison map via the inclusion of $F(\lim X')$ into $F(\lim X)$:
\[
\begin{tikzcd}
 F(\lim X') \ar[r, "\iso"] \ar{d} & \lim F(X') \ni x \ar{d} \\
 F(\lim X) \ar{r} & \lim F(X)  
\end{tikzcd}
\]
For injectivity, suppose $a, b \in F(\lim X)$ map to the same element in $\lim F(X)$, so that $F(\pi_\alpha) a = F(\pi_\alpha) b$ for any $\alpha$. 
Let $X_a' \subseteq \lim X$ be the smallest subset such that $a \in F(X'_a)$ and similarly for $X_b'$. For each $x,y \in X_a'$ 
with $x \neq y$ there is an $\alpha$ 
such that $x_\alpha \neq y_\alpha$. Let $\alpha$ be the upper bound of these as $x,y$ range over all pairs of different 
elements in $X'_a$. Note that this is an element in $\indexord$ because $\indexord$ is regular, and because $X_a$ has cardinality
strictly smaller than some element in $\rho$. Then $\pi_\beta : X'_a \to \pi_\beta[X_a']$ is 
a bijection whenever $\beta \geq \alpha$. Similarly, there is a $\gamma$ such that for all $\beta \geq \gamma$, 
$\pi_\beta : X'_b \to \pi_\beta[X_b']$  is a bijection. Let $\beta$ the least upper bound of $\alpha$ and $\gamma$. 
Then $\pi_\beta[X_a']$ is the smallest subset of $X_\beta$ such that $F(\pi_\beta)(a)$ is in $F(\pi_\beta[X_a'])$,
because a smaller subset $X'' \subset \pi_\beta[X_a']$ would give a subset $\inv{\pi_\beta}(X'') \subset X'_a$
with $a \in F(\inv{\pi_\beta}(X''))$. Similarly, $\pi_\beta[X_b']$ is the smallest subset 
such that $F(\pi_\beta)(b)$ is in $F(\pi_\beta[X_b'])$.
Since $F(\pi_\beta)(a) = F(\pi_\beta)(b)$, we must have $\pi_\beta[X_b'] = \pi_\beta[X_a']$.
So then, $X'_a \iso \pi_\beta[X_a'] \iso X'_b$ via an isomorphism $f$ such that $f(x)_\beta = x_\beta$ from some $\beta$ on. 
We conclude that $X'_a = X'_b$. So $a,b \in F(X'_a)$ map to the same element in $F(\pi_\beta(X'_a))$ by the isomorphism
$F(\pi_\beta)$, and therefore must be equal. 
\end{proof}

\subsection{A Syntactic Condition}
\label{sec:syntactic:condition}

We recall the notion of  $\kappa$-ary algebraic theory for $\kappa$ a regular cardinal. 
A signature consists of a set of operations $\Sigma$ all associated with an arity, which is a $\kappa$-small set. 
Given a signature, one can form the set $\Tm(X)$ of terms with free variables in $X$ inductively in the standard way:
\begin{itemize}
\item If $x \in X$ then $x \in \Tm(X)$
\item If $\op \in \Sigma$ has arity $A$ and $t: A \to \Tm(X)$, then $\op(t) \in \Tm(X)$
\end{itemize}
%If the arity is $0$, the vector in the second case is empty. 
An equation is a pair of terms $(t,s)$ in some $\Tm(X)$. An algebraic 
theory is a signature together with a set of equations over that signature. 

An algebraic theory induces a monad on $\Set$ by the definition
\[
 T(X) = \Tm(X) / \sim
\]
where $\sim$ is the smallest congruence relation generated by the equations of the theory. It is well-known (see e.g.~\cite{adamek1994locally}) 
that the underlying functor of $T$ is then $\kappa$-accessible.  

%
%
%
%Recall that a (finitary) algebraic theory consists of a set of family of operations $\Sigma$ indexed by arities which are just natural numbers, 
%as well as a set of equations, which are just pairs of terms formed from operations in the usual way. An equation $t = s$ is a drop
%equation if the free variables of $t$ are not the same as the free variables of $s$. Given an algebraic theory, one can form the free monad
%generated by it as 
%\[
%  T(X) \defeq \{ t \mid  t \text{ is a term with free variables in } X \} / \sim
%\]
%where $\sim$ is the least congruence relation generated by the equations of the theory. It is well-known REF?? that the free monad 
%generated by an algebraic theory is finitrary, and in fact, any finitary monad is generated by an algebraic theory REF. 

\begin{example}
 The finite powerset monad is generated by the finitary ($\aleph_0$-ary) 
 theory consisting of two operations $\bot$ of arity $\emptyset$ and $\vee$ of arity $\{0,1\}$. 
 We write the latter as an infix binary operator. The equations are 
% together with the equations
\begin{align*}
 x \vee y & = y \vee x & (x \vee y) \vee z & = x \vee (y \vee z) & x \vee x & = x & x \vee \bot & = x
\end{align*}
%This is a finitary algebraic theory.
%None of these equations are drop equations.
\end{example}
\begin{example}
The countable powerset monad is generated by the theory consisting of $\bot$ of arity $\emptyset$ and $\bigvee$ of arity $\N$. 
Following Chapman et al.~\cite{quotientingDelay}, the equations can be expressed by writing $x \vee y$ for 
$\bigvee(x, y,y,y,\dots)$  as
\begin{align*}
x \vee y & = y \vee x & (x \vee y) \vee z & = x \vee (y \vee z) & x \vee x & = x & x \vee \bot & = x \\
x_i \vee (\bigvee x) & = \bigvee x  & (\bigvee x) \vee y & = \bigvee (x_i \vee y)_{i \in \N}
\end{align*}
This is an $\aleph_1$-ary algebraic theory. 
\end{example}

\begin{example}
 The finite distributions monad is generated by a family of binary operations $\oplus_p$ indexed by $p$ in the open interval $(0,1)$
 together with the equations 
\begin{align*}
 x \oplus_p x & = x & x \oplus_p y & = y \oplus_{1-p} x & 
 (x \oplus_p y) \oplus_q z & = x \oplus_{pq} (y \oplus_{\frac{(1-p)q}{1-pq}} z) 
\end{align*}
%This is a finitary algebraic theory.
\end{example}

\begin{example}
 The propositional truncation monad of Example~\ref{ex:prop:trunc} is generated by the algebraic theory of 
 no constructors and a single equation $(x=y)$. 
\end{example}

We now give a syntactic condition for an algebraic theory ensuring that the generated monad satisfies the conditions of Theorem~\ref{thm:semantic:condition}. 

\begin{definition}
An equation $(t,s)$ is a drop equation if $\fv t \neq \fv s$. 
\end{definition}
Here $\fv t$ for $t \in \Tm(X)$ is the smallest subset $X'\subseteq X$ such that $t \in \Tm(X')$.
Of the four examples mentioned above, only propositional truncation has drop equations. 
Note that if a theory has no drop equations, then also the least congruence on terms generated by the equations has no drop equations. 
% neither the theories of finite powerset, nor finite distributions, have drop equations. 
% On the other hand, the propositional truncation monad is generated by a single drop equation. 

%\todo{Update below}

\begin{theorem} \label{thm:syntactic:condition}
 Suppose $T$ is a monad generated by a $\kappa$-ary algebraic theory with no drop equations. Then $T$ preserves  
 pullbacks of monomorphisms and $\kappa$-small intersections. As a consequence, if $\indexord$ is a regular ordinal
 with an $\alpha \in \indexord$ of same cardinality as $\kappa$,  
 %larger than $\kappa$, 
 then $T$ preserves $\opcat{\indexord}$-limits.
\end{theorem}

\begin{proof}
 Let $f : X \to Y$ be a map and let $Z \subseteq Y$. We must show that the following is a pullback
 \[
\begin{tikzcd}
 T(\inv f(Z)) \ar[r, "T(f)"] \ar{d} & T(Z) \ar[d, "T(i)"] \\
 T(X) \ar[r, "T(f)"]  & T(Y)
\end{tikzcd}
 \]
where $i: Z \to Y$ is the inclusion. So suppose given $[s]_Z\in T(Z)$ and $[t]_X\in T(X)$ such that $T(i)[s]_Z = T(f)[t]_X$. Let $t_f$ be
the result of substituting the free variables of $t$ using $f$, so that  $T(f)[t]_X = [t_f]_Y$. The assumption says that 
$t_f$ and $s$ are related in the congruence relation generated by the equalities of the algebraic theory. Since none of these 
are drop equations, $t_f$ and $s$ have the same free variables. Since the free variables of $s$ are in $Z$, the free variables of
$t$ must be in $\inv f(Z)$. So $[t] \in T( \inv f(Z))$ and maps to $[s]_Z\in T(Z)$ and $[t]_X\in T(X)$. This proves existence. 

To prove uniqueness, it suffices to show that $T$ preserves monomorphisms,
implying that the map $T(\inv f(Z)) \to T(X)$ is mono. 
So suppose $A \subseteq B$, 
$j : A \to B$ is the inclusion, and that $t$ and $s$ 
are terms with free variables in $A$ such that $T(j)[s] = T(j)[t]$. 
Since there are no more equations between terms in $B$ than in 
$A$, also $[t] = [s]$.
% Since congruence of terms only equates terms with the same free variables, 
% there are no more equations between terms in $B$ than there are between the same terms in $A$. So $[s] = [t]$. 

Showing that $T$ preserves $\kappa$-small intersections is similar: Given a family $([t_i] \in T(X_i))_{i\in I}$,
with each $X_i \subseteq X$, the free variables in all $t_i$ must be the same, so that they all represent the
same equivalence class in $T(\cap_{i\in I} X_i)$.
\end{proof}

As a non-example, the pullback 
\[
\begin{tikzcd}
 \emptyset \ar{d} \ar{r} \ar[rd, phantom, "\lrcorner"] & \{0\} \ar{d} \\
 \{1\} \ar{r} & \{0,1\}
\end{tikzcd}
\]
witnesses that propositional truncation does not preserve pullbacks of monomorphisms. 

\section{Existential quantification}
\label{sec:existential}

The previous section handled the first requirement listed in Section~\ref{sec:requirements}. We now look at the last two, which both
concern commutativity of universal quantification over clocks and existential quantification. The property we need can be
stated in Clocked Type Theory as % a type equality
\begin{equation} \label{eq:comm:exist}
 \hastype{\Delta, X : \univ\Delta, \phi : X \to \forall\kappa . \prop{\Delta,\kappa}}{\exists (x : X). \forall\kappa. \capp{\phi(x)} = \forall\kappa. \exists (x:X). \capp{\phi(x)}}{\prop\Delta}
\end{equation}
where $\Delta = \kappa_1 : \clocktype, \dots, \kappa_n : \clocktype$ is a clock context of length $n$. We first note that this principle
does not hold in general.

\begin{example}
 Let $\Delta$ be the empty context and let $X = \indexord$ be the indexing ordinal. Let $\phi$ be the predicate
 on $\indexord\times \clk$ defined as
 \[
 \phi\quadruple{\FSA}\vartheta{\alpha}\lambda = \{\star \mid \vartheta(\lambda) \leq \alpha \}
 \]
 Since this predicate is invariant under clock introduction, it corresponds to a map $X \to \forall\kappa . \prop{\Delta,\kappa}$, but
 $\forall\kappa . \exists(\alpha : \indexord) . \capp{\phi(\alpha)} = \{\star\}$ 
 and 
\begin{align*}
 (\exists(\alpha : \indexord) . \forall\kappa . \capp{\phi(\alpha)})\timeobj{\FSA}\vartheta 
 & = \{\star \mid \exists \alpha : \indexord . \forall\beta . \phi\quadruple{(\FSA, \lambda_\FSA)}{\vartheta[\lambda_\FSA \mapsto \beta]}\alpha{\lambda_\FSA} \} \\
 & = \emptyset
\end{align*}
 because the condition fails when $\beta > \alpha$. 
 Such $\beta$ always exist, because $\indexord$ is assumed to be a limit ordinal, as it needs to be to interpret axiom (\ref{eq:force}). 
%\begin{align*}
% \forall\kappa . \exists(\alpha : \indexord) . \phi(\alpha) & = \{\star\}  & 
% \exists(\alpha : \indexord) . \forall\kappa . \phi(\alpha) & = \emptyset
%\end{align*}
\end{example}

The above example is a variant of the well-known phenomenon that existential quantification in the internal language of a presheaf 
topos is local and does not imply global existence. This has been observed previously also in the setting of guarded
recursion~\cite{ToT,TransfiniteIris}. 

To understand the principle more generally, we look at substitutions from some object $\Gamma$ in $\grtotal$ into the context of (\ref{eq:comm:exist}). 
A mapping $\Gamma \to \clk^\Delta$ corresponds to an object of $\slice \grtotal{\clk^\Delta}$ which is equivalent to $\gr\Delta$, so
a mapping into the context of (\ref{eq:comm:exist}) corresponds to an object $\Gamma$ in $\gr\Delta$, a family $X$ over $\Gamma$, 
and finally, a predicate $\phi$ on $X \times \clk$. The equality (\ref{eq:comm:exist})
then translates to equality of 
$(\exists (x : X). \forall\kappa. \phi(x)) \triple{(\FSA,\vartheta)}{\chi}\gamma$ which is
\begin{align} \label{eq:exists:forall}
 \{ \star \mid \exists x : X\triple{(\FSA, \vartheta)}{\chi}\gamma. \forall 
 \alpha < \indexord . \phi\sixtuple{(\FSA, \lambda_\FSA)}{\vartheta[\lambda_\FSA\mapsto \alpha]}{\iota\chi}{\iota\cdot\gamma}{\iota\cdot x}{\lambda_\FSA} \}
\end{align}
and $(\forall\kappa.\exists (x : X). \phi(x)) \quadruple{\FSA}\vartheta{\chi}\gamma$ which is 
\begin{align} \label{eq:forall:exists}
 \{ \star \mid \forall 
 \alpha < \indexord . \exists x : X\quadruple{(\FSA, \lambda_\FSA)}{\vartheta[\lambda_\FSA\mapsto \alpha]}{\iota\chi}{\iota\cdot\gamma}. \phi\sixtuple{(\FSA, \lambda_\FSA)}{\vartheta[\lambda_\FSA\mapsto \alpha]}{\iota\chi}{\iota\cdot\gamma}{x}{\lambda_\FSA} \}
\end{align} 
for all $\gamma \in \Gamma\triple{\FSA}\vartheta{\chi}$ and $\chi : \Delta \to \FSA$. 

\begin{theorem} \label{thm:comm:exists:small}
 If $\indexord$ is a regular ordinal and all fibres $X\quadruple{\FSA}\vartheta{\chi}\gamma$ are of cardinality strictly smaller than $\indexord$, then 
 (\ref{eq:exists:forall}) and (\ref{eq:forall:exists}) are equal.
\end{theorem}

\begin{proof}
 Since $X$ is invariant under clock introduction, (\ref{eq:forall:exists}) can be rewritten as
 \begin{equation}\label{eq:forall:exists:2}
  \{ \star \mid \forall 
 \alpha < \indexord . \exists x : X\quadruple{\FSA}{\vartheta}{\chi}{\gamma}. \phi\sixtuple{(\FSA, \lambda_\FSA)}{\vartheta[\lambda_\FSA\mapsto \alpha]}{\iota\chi}{\iota\cdot\gamma}{\iota\cdot x}{\lambda_\FSA} \}
\end{equation}
Clearly the condition of (\ref{eq:exists:forall}) implies that of (\ref{eq:forall:exists:2}). 
For the other direction, first note that the condition is downward closed in $\alpha$ because, for $\beta\leq \alpha$, the identity on 
$(\FSA, \lambda_\FSA)$ tracks a map $\sigma$ % : \timeobj{(\FSA, \lambda_\FSA)}{\vartheta[\lambda_\FSA\mapsto \alpha]} 
%\to $ in $\catT$ 
which gives 
\[
 \sigma \cdot(-) : \phi\sixtuple{(\FSA, \lambda_\FSA)}{\vartheta[\lambda_\FSA\mapsto \alpha]}{\iota\chi}{\iota\cdot\gamma}{\iota\cdot x}{\lambda_\FSA} \to \phi\sixtuple{(\FSA, \lambda_\FSA)}{\vartheta[\lambda_\FSA\mapsto \beta]}{\iota\chi}{\iota\cdot\gamma}{\iota\cdot x}{\lambda_\FSA}
\]
Now, arguing classically, suppose (\ref{eq:exists:forall}) is false. %Since the condition is downward closed in $\alpha$, f
For each $x$ there must exist  an $\alpha_x$, such that
\[
  \phi\sixtuple{(\FSA, \lambda_\FSA)}{\vartheta[\lambda_\FSA\mapsto \beta]}{\iota\chi}{\iota\cdot\gamma}{\iota\cdot x}{\lambda_\FSA} 
\]
is false for all $\beta\geq\alpha_x$. By the condition of the cardinality of the fibres of $X$ and since $\indexord$ is regular, the least upper bound
$\alpha$ of all the $\alpha_x$ is an ordinal which is smaller than $\indexord$. This ordinal then proves (\ref{eq:forall:exists}) false.
\end{proof}

The above proof is essentially the same as the proof of the `existential property'  forming part of the basis 
of Transfinite Iris~\cite[Theorem~5.2]{TransfiniteIris}, but adapted to the multi-clocked setting. 
The existential property states that any existential proved in
the logic of Transfinite Iris externalises to an existential in the meta-logic. Transfinite Iris uses a large indexing ordinal
to ensure that the indexing ordinal is larger than the existentially quantified type. We will discuss this point further 
in the conclusion. 

Universal quantification over clocks also commutes with essentially unique existential quantifications.

\begin{theorem} \label{thm:comm:exists:unique}
 Let $\Gamma$ be an object of $\gr\Delta$, 
 $X$ a family over $\Gamma$, and $\phi$ a predicate over $X \times \clk$, and let $n : \Gamma \to \N$
 be a map of presheaves. If $\phi(x, \kappa) \wedge \phi(y,\kappa)$ implies $(\later^\kappa)^n(x=y)$ then (\ref{eq:exists:forall}) and (\ref{eq:forall:exists}) are equal.
% \todo{Terminology: Family is type invariant under clock intro?}
\end{theorem}

\begin{proof}
 The predicate 
 \[
   \genericjudg{\Delta,\Gamma, x : X, y: X, \kappa : \clocktype}{(\later^\kappa)^n(x=y)}
 \] 
 is true at $\seventuple\FSA\vartheta\chi \gamma xy\lambda$ if and only if either $x=y$ or $\vartheta(\lambda)<n(\gamma)$. 
 So, for $\alpha$ larger than $n(\gamma)$, there exists at most a single $x$ such that
\[
  \phi\sixtuple{(\FSA, \lambda_\FSA)}{\vartheta[\lambda_\FSA\mapsto \alpha]}{\iota\chi}{\iota\cdot\gamma}{\iota\cdot x}{\lambda_\FSA} 
\]
Suppose now that (\ref{eq:forall:exists:2}) holds, and pick $\alpha$ larger than $n(\gamma)$. By (\ref{eq:forall:exists:2}) there is an $x$ such that
the above holds for that $\alpha$ and $x$. By downward closedness, it then also holds at all $\beta\leq \alpha$. For $\beta\geq\alpha$, there is
a $y$ such that the above holds at $\beta$ for that $y$. By downwards closure and uniqueness, we then conclude $x=y$, so that $x$ witnesses 
(\ref{eq:exists:forall}).
%
%Since the predicate is downwards closed in $\alpha$, and $x$ does not depend on $\alpha$, if the predicates is true for some $x$ and 
%some $\alpha$ larger than $\gamma(n)$ it must hold for the same $x$ for all $\alpha$. So (\ref{eq:forall:exists}) then implies (\ref{eq:exists:forall}).
\end{proof}

%\todo{Terminology: Family or type? Predicate?}

%\section{Results}
%\label{sec:results}
%
%\todo{Examples of results proved in previous papers that now hold in $\Set$. Maybe skip if I do not have the time to inspect proofs?}

\section{Conclusion}
\label{sec:conclusion}

We have seen that by choosing the indexing ordinal of the model to be sufficiently big, all the requirements of the applications 
studied in Section~\ref{sec:applications} can be met. For example, to program with finitely branching non-deterministic processes, and reasoning about
bisimilarity of these using guarded recursion, as done by Kristensen et al~\cite{CubicalCloTT} the indexing ordinal can be $\omega_1$
(the first uncountable ordinal) or higher. For modelling Probabilistic FPC and reasoning about contextual equivalence of programs,
the indexing ordinal needs to be larger than the cardinality of $\Dfin(\Dforall 1)$, which is uncountable. Finally, for reasoning about the coinductive delay
monad up to weak bisimilarity as a monad on the category of setoids, one needs the indexing monad to be $\omega_1$ or higher, such that
quantification over clocks commutes with existential quantification over $\N$. In all these examples, we can now interpret the results proved in
previous papers into the model and conclude the expected set theoretic results. For example, Theorem~7.7 of~\cite{FewExtraPages} now proves the existence
of a distributive law of the monad generated by an algebraic theory with no drop equations over the coinductive delay monad up to weak bisimilarity
as monads on the usual category of setoids. Likewise the results about contextual equivalence of Probabilistic FPC programs proved in~\cite{POPL25} 
now hold for the operational semantics given in that paper as interpreted in the usual set-theoretic way, where $\Dforall$ 
is a coinductively defined monad.

The syntactic requirements on algebraic theories to generate monads commuting with clock quantification are quite natural from a 
type theoretic point of view: The arities should be of a particular form and the equations non-drop. The latter corresponds to the equations being 
typeable in a relevant type system (in the technical sense of substructural types). 
The requirements on cardinality for existential quantifications, on the other hand, are less 
natural in type theory. One way to avoid this is to use the ordinal
of all ordinals as indexing, as done in Transfinite Iris~\cite{TransfiniteIris}. 
Since this ordinal is large (lives in the next universe), it suffices to commute 
clock quantification with all small existential quantifications. 
%Such requirements of universe levels would be more natural in type theory
%than requirements of cardinality. 
This simplification comes at a price, however: Universal quantification over clocks now raises the universe 
level for small types. Monads such as $\Dforall$ defined using universe quantification will therefore only be monads on the category 
of sets in the next universe. In particular, the type $\Dforall 1$ will not be small, 
and so existential quantification over it will not
automatically commute with universal quantification over clocks, 
unless one can prove $\Dforall 1$ isomorphic to a small set. 

% Future work includes extending Theorem~\ref{thm:semantic:condition} beyond the finitary case to more generally accessible functors. 
% Most of the proof of Theorem~\ref{thm:semantic:condition} can be generalised if the functor is assumed to preserve sufficiently large
% intersections. The problem is the argument that proves that the diagram $X'$ must consist of bijections from some point. This argument 
% does not generalise from diagrams of finite sets to sets of larger cardinality. 
% It would also be good to extend beyond $\Set$, to slices of $\Set$ or, more generally, presheaves. This would allow more general 
% quotient inductive types to be considered. 

An alternative to proving type constructors commute with clock quantification is to use Worrels proof~\cite{Worrell05} that for 
any $\kappa$-accessible endofunctor $F$ on $\Set$, the terminal sequence converges in $\kappa + \kappa$ steps, producing 
a terminal coalgebra for $F$. This proves that the encoding of coinductive types for those functors is correct, externally. While this
result is more general (in the sense that it works for all cardinals $\kappa$), it is also less natural and flexible from the type theoretic
point of view. From the computational point of view, Kristensen et al.~\cite{CubicalCloTT} develop a notion of induction under clocks
that provides computational content to axioms of higher-inductive types or quotient inductive types commuting with clock quantification.
It is not clear how to provide computational content to axioms of being a terminal coalgebra. 

Multimodal Type Theory~\cite{gratzer2020multimodal} gives a different approach to reasoning about the relationship between $\Set$ 
and the topos of trees in one type theory. Each of these categories correspond to a mode of the type theory, and the functors of
constant presheaves and global elements are internalised as modalities between these two modes. While the use of modes
makes the relationship between these categories explicit in the language, as we saw in Section~\ref{sec:model} the same two
models form submodels of the model of multiclocked guarded recursion. The results of Sections~\ref{sec:alg:types} and~\ref{sec:existential}
can also be applied in the setting of Multimodal Type Theory. 

%A similar construction of 
%a subdiagram appears in~\cite[Proof of Thm 3.3]{adamek2024finitary}, where the subdiagram is shown to 

%\todo{Generalising to higher ordinals for algebraic theories, mention Worrell proof.}
%
%
%\subsection{Related Work}
%
%Higher indexing guarded recursion has been considered before. For example, the topos of trees model is generalised from presheaves on $\omega$
%to sheaves over a complete Heyting algebra with a well-founded base. Applying this to the ideal completion of an ordinal gives presheaves on that ordinal.
%In the multiclocked case, both Bizjak and M{\o}gelberg REF and Mannaa et al REF restrict themselves to indexing over $\omega$, although the constructions
%directly generalise. % Palombi and Sterling~\cite{palombi-sterling-2023}
%
%Bizjak et al REF use sheaves indexed over $\omega_1$ to model must-contextual equivalence of non-deterministic processes. Transfinite Iris
%uses higher ordinal indexing in order to obtain an existential property similar to 
%
%The topos of trees 
%
%\todo{Clocked Cubical computes}
%
%\subsection{Future Work}
%
%%\section{Discussion (Related work?)}
%%\label{sec:discussion}

%
%\todo{Higher ordinals also considered elsewhere (topos of trees, Jon, existential dilemma.)}

%\todo{Multimodal type theory also allows talking about externalisation of results}

%% Please use bibtex, 

\bibliography{paper}

@article{FewExtraPages,
  title={What Monads Can and Cannot Do with a Few Extra Pages},
  author={M{\o}gelberg, Rasmus Ejlers and Zwart, Maaike},
  journal={Logical Methods in Computer Science},
  volume={21},
  year={2025},
  publisher={Episciences. org}
}

@article{POPL25,
  author       = {Philipp Stassen and
                  Rasmus Ejlers M{\o}gelberg and
                  Maaike Zwart and
                  Alejandro Aguirre and
                  Lars Birkedal},
  title        = {Modelling Recursion and Probabilistic Choice in Guarded Type Theory},
  journal      = {Proc. {ACM} Program. Lang.},
  volume       = {9},
  number       = {{POPL}},
  pages        = {1417--1445},
  year         = {2025},
  url          = {https://doi.org/10.1145/3704884},
  doi          = {10.1145/3704884},
  bibsource    = {dblp computer science bibliography, https://dblp.org}
}

@article{atkey13icfp,
  title={Productive coprogramming with guarded recursion},
  author={Atkey, Robert and McBride, Conor},
  journal={ACM SIGPLAN Notices},
  volume={48},
  number={9},
  pages={197--208},
  year={2013},
  publisher={ACM New York, NY, USA}
}

@inproceedings{GDTT,
  title={Guarded dependent type theory with coinductive types},
  author={Bizjak, Ale{\v{s}} and Grathwohl, Hans Bugge and Clouston, Ranald and M{\o}gelberg, Rasmus E and Birkedal, Lars},
  booktitle={International Conference on Foundations of Software Science and Computation Structures},
  pages={20--35},
  year={2016},
  organization={Springer}
}

@Book{hottbook,
  author =    {The {Univalent Foundations Program}},
  title =     {Homotopy Type Theory: Univalent Foundations of Mathematics},
  address =   {Institute for Advanced Study},
  publisher = {\url{https://homotopytypetheory.org/book}},
  year =      2013}

@article{ToT,
  title={First steps in synthetic guarded domain theory: step-indexing in the topos of trees},
  author={Lars Birkedal and Rasmus Ejlers M{\o}gelberg and Jan Schwinghammer and Kristian St{\o}vring},
  journal={Logical Methods in Computer Science},
  volume={8},
  number={4},
  year={2012}
}

@InProceedings{CTT,
  title={Cubical Type Theory: A Constructive Interpretation of the Univalence Axiom},
  author={Cohen, Cyril and Coquand, Thierry and Huber, Simon and M{\"o}rtberg, Anders},
  booktitle={21st International Conference on Types for Proofs and Programs (TYPES 2015)},
  year={2018},
  organization={Schloss Dagstuhl-Leibniz-Zentrum fuer Informatik}
}

@inproceedings{Mogelberg14,
  title={A type theory for productive coprogramming via guarded recursion},
  author={M{\o}gelberg, Rasmus Ejlers},
  booktitle={Proceedings of the Joint Meeting of the Twenty-Third EACSL Annual Conference on Computer Science Logic (CSL) and the Twenty-Ninth Annual ACM/IEEE Symposium on Logic in Computer Science (LICS)},
  pages={1--10},
  year={2014}
}

@article{GDTTmodel,
  author    = {Ale\v{s} Bizjak and Rasmus Ejlers M{\o}gelberg},
title={Denotational semantics for guarded dependent type theory},
volume={30},
number={4},
journal={Mathematical Structures in Computer Science},
publisher={Cambridge University Press},
year={2020},
pages={342--378}
}

@inproceedings{bahr2017clocks,
  title={The clocks are ticking: No more delays!},
  author={Patrick Bahr. and Hans Bugge Grathwohl and Rasmus Ejlers M{\o}gelberg},
    booktitle={2017 32nd Annual ACM/IEEE Symposium on Logic in Computer Science (LICS)},
  pages={1--12},
  year={2017},
  organization={IEEE}
}

@article{drat,
  author    = {Ranald Clouston and
               Bassel Mannaa and
               Rasmus Ejlers M{\o}gelberg and
               Andrew M. Pitts and
               Bas Spitters},
  title     = {Modal Dependent Type Theory and Dependent Right Adjoints},
  volume={30},
  number={2},
  journal={Mathematical Structures in Computer Science},
  publisher={Cambridge University Press},
  year      = {2020},
  pages={118--138}
}

@unpublished{Hofmann-Streicher:lifting,
  title={Lifting {G}rothendieck universes},
  author={Hofmann, Martin and Streicher, Thomas},
  journal={Unpublished manuscript},
  year={1999},
  note = {Unpublished},
  url = {www.mathematik.tu-darmstadt.de/~streicher/NOTES/lift.pdf}
}

@inproceedings{clouston2018fitch,
  title={Fitch-style modal lambda calculi},
  author={Clouston, Ranald},
  booktitle={International Conference on Foundations of Software Science and Computation Structures},
  pages={258--275},
  year={2018},
  organization={Springer}
}

@article{CubicalAgda,
  title={Cubical Agda: a dependently typed programming language with univalence and higher inductive types},
  author={Vezzosi, Andrea and M{\"o}rtberg, Anders and Abel, Andreas},
  journal={Proceedings of the ACM on Programming Languages},
  volume={3},
  number={ICFP},
  pages={1--29},
  year={2019},
  publisher={ACM New York, NY, USA}
}

@article{clottmodel,
  title = {{Ticking clocks as dependent right adjoints: Denotational semantics for
  clocked type theory}},
  author = {Mannaa, Bassel and M{\o}gelberg, Rasmus Ejlers and Veltri, Niccol{\`o}},
  journal={Logical Methods in Computer Science},
  volume={16},
  year={2020},
  publisher={Episciences. org}
}

@article{quotientingDelay,
  title={Quotienting the delay monad by weak bisimilarity},
  author={Chapman, James and Uustalu, Tarmo and Veltri, Niccol{\`o}},
  journal={Mathematical Structures in Computer Science},
  volume={29},
  number={1},
  pages={67--92},
  year={2019}
}

@inproceedings{GuardedCubicalAgda,
author = {Veltri, Niccol\`{o} and Vezzosi, Andrea},
title = {Formalizing $\pi$-Calculus in Guarded Cubical Agda},
year = {2020},
isbn = {9781450370974},
publisher = {Association for Computing Machinery},
address = {New York, NY, USA},
url = {https://doi.org/10.1145/3372885.3373814},
doi = {10.1145/3372885.3373814},
booktitle = {Proceedings of the 9th ACM SIGPLAN International Conference on Certified Programs and Proofs},
pages = {270–283},
numpages = {14},
location = {New Orleans, LA, USA},
series = {CPP 2020}
}

@article{Worrell05,
  author    = {James Worrell},
  title     = {On the final sequence of a finitary set functor},
  journal   = {Theor. Comput. Sci.},
  volume    = {338},
  number    = {1-3},
  pages     = {184--199},
  year      = {2005},
  url       = {https://doi.org/10.1016/j.tcs.2004.12.009},
  doi       = {10.1016/j.tcs.2004.12.009},
  bibsource = {dblp computer science bibliography, https://dblp.org}
}

@inproceedings{gratzer2020multimodal,
  title={Multimodal dependent type theory},
  author={Gratzer, Daniel and Kavvos, GA and Nuyts, Andreas and Birkedal, Lars},
  booktitle={Proceedings of the 35th Annual ACM/IEEE Symposium on Logic in Computer Science},
  pages={492--506},
  year={2020}
}

@inproceedings{mogelberg2021two,
  author    = {Rasmus Ejlers M{\o}gelberg and
               Andrea Vezzosi},
  editor    = {Ana Sokolova},
  title     = {Two Guarded Recursive Powerdomains for Applicative Simulation},
  booktitle = {Proceedings 37th Conference on Mathematical Foundations of Programming
               Semantics, {MFPS} 2021, Hybrid: Salzburg, Austria and Online, 30th
               August - 2nd September, 2021},
  series    = {{EPTCS}},
  volume    = {351},
  pages     = {200--217},
  year      = {2021},
  url       = {https://doi.org/10.4204/EPTCS.351.13},
  doi       = {10.4204/EPTCS.351.13},
  bibsource = {dblp computer science bibliography, https://dblp.org}
}

@inproceedings{CubicalCloTT,
  author       = {Magnus Baunsgaard Kristensen and
                  Rasmus Ejlers M{\o}gelberg and
                  Andrea Vezzosi},
  editor       = {Christel Baier and
                  Dana Fisman},
  title        = {Greatest HITs: Higher inductive types in coinductive definitions via
                  induction under clocks},
  booktitle    = {{LICS} '22: 37th Annual {ACM/IEEE} Symposium on Logic in Computer
                  Science, Haifa, Israel, August 2 - 5, 2022},
  pages        = {42:1--42:13},
  publisher    = {{ACM}},
  year         = {2022},
  url          = {https://doi.org/10.1145/3531130.3533359},
  doi          = {10.1145/3531130.3533359},
  bibsource    = {dblp computer science bibliography, https://dblp.org}
}

@article{adamek1974free,
  title={Free algebras and automata realizations in the language of categories},
  author={Ad{\'a}mek, Ji{\v{r}}{\'\i}},
  journal={Commentationes Mathematicae Universitatis Carolinae},
  volume={15},
  number={4},
  pages={589--602},
  year={1974},
  publisher={Charles University in Prague, Faculty of Mathematics and Physics}
}

@article{adamek1995greatest,
  title={On the greatest fixed point of a set functor},
  author={Ad{\'a}mek, Ji{\v{r}}{\'\i} and Koubek, V{\'a}clav},
  journal={Theoretical Computer Science},
  volume={150},
  number={1},
  pages={57--75},
  year={1995},
  publisher={Elsevier}
}

@book{adamek1994locally,
  title={Locally Presentable and Accessible Categories},
  author={Ad{\'a}mek, , Ji{\v{r}}{\'\i} and Rosicky, , Ji{\v{r}}{\'\i}},
  volume={189},
  year={1994},
  series = {London Mathematical Society Lecture Note Series},
  publisher={Cambridge University Press}
}

@article{adamek2019finitary,
  title={ON FINITARY FUNCTORS},
  author={Ad{\'a}mek, Ji{\v{r}}{\'\i} and Milius, Stefan and Sousa, Lurdes and Wissmann, Thorsten},
  journal={Theory and Applications of Categories},
  volume={34},
  number={35},
  pages={1134--1164},
  year={2019}
}

@inproceedings{TransfiniteIris,
  title={Transfinite Iris: resolving an existential dilemma of step-indexed separation logic},
  author={Spies, Simon and G{\"a}her, Lennard and Gratzer, Daniel and Tassarotti, Joseph and Krebbers, Robbert and Dreyer, Derek and Birkedal, Lars},
  booktitle={Proceedings of the 42nd ACM SIGPLAN International Conference on Programming Language Design and Implementation},
  pages={80--95},
  year={2021}
}

@article{GiovanniniDN25,
  author       = {Eric Giovannini and
                  Tingting Ding and
                  Max S. New},
  title        = {Denotational Semantics of Gradual Typing using Synthetic Guarded Domain
                  Theory},
  journal      = {Proc. {ACM} Program. Lang.},
  volume       = {9},
  number       = {{POPL}},
  pages        = {772--801},
  year         = {2025},
  url          = {https://doi.org/10.1145/3704863},
  doi          = {10.1145/3704863},
  bibsource    = {dblp computer science bibliography, https://dblp.org}
}

\end{document}